\documentclass[
paper=letter,%
numbers=noendperiod,%
captions=nooneline,%
DIV=13%
]{scrartcl}

\usepackage[T1]{fontenc}%
\usepackage{lmodern}%
\usepackage[american]{babel}%
\usepackage{microtype}%

\usepackage[
hyperref,%
table%
]{xcolor}%

\usepackage{scrlayer-scrpage}%

\usepackage[latin1,utf8]{inputenc}
\usepackage{dsfont} %
\usepackage{mathrsfs}
\usepackage{tipa}%
\usepackage[colorinlistoftodos,prependcaption,textsize=tiny]{todonotes}
\usepackage{rotating}%
\usepackage{amsmath}%
\usepackage{mathtools}%
\usepackage{amssymb}%
\usepackage{amsthm}%
\usepackage{thmtools}%
\usepackage{etoolbox}%
\usepackage{bm}%
\usepackage{bbm}%
\usepackage{enumitem}%
\usepackage{graphicx}%
\usepackage{grffile}%
\usepackage{tikz}%
\usepackage{tabularx}%
\usepackage{siunitx}%
\usepackage{booktabs}%
\usepackage{vruler}%
\usepackage{csquotes}%
\usepackage[
style=authoryear,%
dashed=false,%
hyperref=true,%
useprefix=false,%
maxnames=2,%
uniquelist=minyear,%
maxbibnames=6,%
uniquename=false,%
seconds=true,%
date=iso,%
urldate=iso%
]{biblatex}%
\AtEveryCitekey{\toggletrue{blx@useprefix}}
\DeclareAutoCiteCommand{inlinetext}{\textcite}{\textcites}
\usepackage[
hypertexnames=false,%
setpagesize=false,%
pdfborder={0 0 0},%
pdfstartview=Fit,%
bookmarksopen=true,%
bookmarksnumbered=true%
]{hyperref}%
\usepackage[hyphens]{xurl}
\hypersetup{breaklinks=true}
\setkomafont{pageheadfoot}{\normalfont\normalcolor\sffamily}%
\setkomafont{pagenumber}{\normalfont\normalcolor\sffamily}%
\automark{section}%
\setkomafont{captionlabel}{\normalfont\normalcolor\sffamily\bfseries}%

\newcommand*{\mysquare}{\rule[0.18em]{0.36em}{0.36em}}
\newcommand*{\mytriangle}{\raisebox{0.12em}{\resizebox{0.48em}{0.48em}{$\blacktriangleright$}}}
\newcommand*{\mybar}{\rule[0.32em]{0.62em}{0.08em}}
\newcommand*{\mydot}{\raisebox{0.14em}{\resizebox{0.44em}{!}{$\bullet$}}}
\setlist{%
  align=left,%
  labelindent=0mm, %
  leftmargin=!,%
  itemindent=0mm, %
  listparindent=\parindent,%
  parsep=0mm,%
  topsep=1mm,%
  itemsep=1mm%
}
\setlist[itemize,1]{label={\mysquare\ }, labelwidth=\widthof{\mysquare\ }}%
\setlist[itemize,2]{label={\mytriangle\ }, labelwidth=\widthof{\mytriangle\ }}%
\setlist[itemize,3]{label={\mybar\ }, labelwidth=\widthof{\mybar\ }}%
\setlist[itemize,4]{label={\mydot\ }, labelwidth=\widthof{\mydot\ }}%
\setlist[enumerate,1]{label=\arabic*), labelwidth=\widthof{9)}}%
\setlist[enumerate,2]{label=\arabic{enumi}.\arabic*), labelwidth=\widthof{9.9)}}%
\setlist[enumerate,3]{label=\arabic{enumi}.\arabic{enumii}.\arabic*), labelwidth=\widthof{9.9.9)}}%

\newcommand*{\abstractnoindent}{}%
\let\abstractnoindent\abstract
\renewcommand*{\abstract}{\let\quotation\quote\let\endquotation\endquote
  \abstractnoindent}
\deffootnote[1em]{1em}{1em}{\textsuperscript{\thefootnotemark}}%
\renewcommand*{\big}[1]{{\vcenter{\hbox{\scalebox{1.30}{\ensuremath#1}}}}}%

\definecolor{blue}{RGB}{58, 95, 205}%
\definecolor{red}{RGB}{205, 41, 144}%
\definecolor{orange}{RGB}{238, 118, 0}%
\definecolor{chocolate}{RGB}{205, 102, 29}%

\DeclareNameAlias{sortname}{family-given}%
\DefineBibliographyExtras{american}{\DeclareQuotePunctuation{}}%
\renewbibmacro*{volume+number+eid}{%
  \setunit*{\addcomma\space}%
  \printfield{volume}%
  \printfield{number}}
\DeclareFieldFormat*{number}{(#1)}
\DeclareFieldFormat*{title}{#1}%
\DeclareFieldFormat{doi}{%
\ifhyperref
  {\href{http://dx.doi.org/#1}{\nolinkurl{doi:#1}}}%
  {\nolinkurl{doi:#1}}}%
\renewbibmacro*{in:}{}%
\DeclareFieldFormat{isbn}{ISBN #1}%
\DeclareFieldFormat{pages}{#1}%
\DeclareFieldFormat{url}{\url{#1}}%
\DeclareFieldFormat{urldate}{\mkbibparens{#1}}%
\renewcommand*{\cite}[2][]{\textcite[#1]{#2}}

\newif\ifstarttheorem
\declaretheoremstyle[%
spaceabove=0.5em,
spacebelow=0.5em,
headfont=\sffamily\bfseries\global\starttheoremtrue,
notefont=\sffamily\bfseries,
notebraces={(}{)},
headpunct={},
bodyfont=\normalfont,
postheadspace=\newline,
]{myProofLessStyle}
\declaretheorem[style=myProofLessStyle, numberwithin=section]{definition}%

\declaretheorem[style=myProofLessStyle, sibling=definition]{corollary}
\declaretheorem[style=myProofLessStyle, sibling=definition]{remark}
\declaretheorem[style=myProofLessStyle, sibling=definition]{example}
\declaretheorem[style=myProofLessStyle, sibling=definition]{algorithm}

\newif\ifstarttheorem
\declaretheoremstyle[%
spaceabove=0.5em,
spacebelow=0.5em,
headfont=\sffamily\bfseries\global\starttheoremtrue,
notefont=\sffamily\bfseries,
notebraces={(}{)},
headpunct={},
bodyfont=\normalfont,
postheadspace=\newline%
]{myMainStyle}
\declaretheorem[style=myMainStyle, sibling=definition]{proposition}
\declaretheorem[style=myMainStyle, sibling=definition]{lemma}
\declaretheorem[style=myMainStyle, sibling=definition]{theorem}

\makeatletter
\preto\itemize{%
  \if@inlabel
    \ifstarttheorem
      \mbox{}\par\nobreak\vskip\glueexpr-\parskip-\baselineskip+0.25em\relax\hrule\@height\z@
    \fi%
  \fi%
  \global\starttheoremfalse%
  \def\tempa{proof}%
  \ifx\tempa\mycurrenvir
    \ifstarttheorem
      \mbox{}\par\nobreak\vskip\glueexpr-\parskip-\baselineskip+0.25em\relax\hrule\@height\z@
    \fi%
  \fi%
  \global\starttheoremfalse%
}
\preto\enditemize{\global\starttheoremfalse}
\makeatother

\makeatletter
\preto\enumerate{%
  \if@inlabel
    \ifstarttheorem
      \mbox{}\par\nobreak\vskip\glueexpr-\parskip-\baselineskip+0.25em\relax\hrule\@height\z@
    \fi%
  \fi%
  \global\starttheoremfalse%
  \def\tempa{proof}%
  \ifx\tempa\mycurrenvir
    \ifstarttheorem
      \mbox{}\par\nobreak\vskip\glueexpr-\parskip-\baselineskip+0.25em\relax\hrule\@height\z@
    \fi%
  \fi%
  \global\starttheoremfalse%
}
\preto\endenumerate{\global\starttheoremfalse}
\makeatother

\def\IR{\mathbb{R}}
\def\IN{\mathbb{N}}
\def\Prob{\mathbb{P}}
\def\E{\mathbb{E}}
\def\bfu{\bm{u}}
\def\bfx{\bm{x}}
\def\bfU{\bm{U}}
\def\bfX{\bm{X}}

\def\d{\,\mathrm{d}}

\newcommand{\abs}[1]{\left\vert #1 \right\vert}

\newcommand{\ceil}[1]{\left\lceil#1\right\rceil}

\def\gen{\psi}
\def\invgen{\gen^{-1}}

\def\Copula{C}

\def\an{a}
\def\bn{b}
\def\ban{\bm{\an}}
\def\bbn{\bm{\bn}}

\def\jointF{F}
\def\pMargin{F}

\def\pD{D}
\def\dD{\pD'}
\def\qD{\pD^{-1}}
\def\qDs{\pD^{*\,-1}}

\def\pF{F}

\def\EVDist{H}

\def\STDF{\ell}

\def\bfM{\bm{M}}

\def\JointDist{\pMargin}

\DeclareMathOperator{\lpois}{Pois}

\definecolor{todocolor}{RGB}{255,165,0}

\newcommand*{\U}{\operatorname{U}}

\newcommand*{\omu}[3]{\underset{#3}{\overset{#1}{#2}}}
\newcommand*{\isim}{\omu{\text{\tiny{ind.}}}{\sim}{}}

\newcommand*{\supp}{\operatorname{supp}}

\DeclareMathOperator{\RV}{RV}
\DeclareMathOperator{\MDA}{MDA}

\def\rate{r}
\DeclareMathOperator{\support}{supp}
\DeclareMathOperator{\GEV}{GEV}

\newcommand{\ind}[1]{\mathds{1}_{#1}}

\definecolor{todocolor}{RGB}{155,0,155}
\definecolor{todocolor}{RGB}{200,64,0}

\definecolor{links}{RGB}{0,0,128}

\newcommand\ct[1]{\text{\rmfamily\upshape #1}}
\def\e{\ct{e}}

\def\bfa{\bm{a}}
\def\bfb{\bm{b}}

\def\EW{\Psi}

\def\SetCopula{\mathscr{C}}
\def\SetEVTCopula{\mathscr{C}^{\text{EVC}}}

\DeclareMathOperator{\ZetaDist}{Zeta}
\DeclareMathOperator{\polylog}{Li}

\newcommand\AMDCO[1]{\operatorname{AMD}_{#1}^1}

\def\id{\mathrm{id}}

\def\GH{\mathrm{GH}}
\def\GHCopula{\Copula^{\GH}}

\def\Cl{\mathrm{Cl}}
\def\ClCopula{\Copula^{\Cl}}

\DeclareMathOperator{\vol}{\Delta}

\def\SetCMonDist{\mathrm{D}_d^{\Psi}}

\begin{document}

\thispagestyle{plain}
\begin{center}
  \sffamily
  {\bfseries\LARGE Morillas-type transformations of copulas and stable tail dependence functions\par} %
  \bigskip\smallskip {\Large Klaus Herrmann\footnote{Département de
      Mathématiques, Université de Sherbrooke, \href{mailto:klaus.herrmann@usherbrooke.ca}{\nolinkurl{klaus.herrmann@usherbrooke.ca}}.},
    Marius Hofert\footnote{Department of Statistics and Actuarial Science, The
      University of Hong Kong,
      \href{mailto:mhofert@hku.hk}{\nolinkurl{mhofert@hku.hk}}.},
    Mélina Mailhot\footnote{Department of Mathematics and Statistics, Concordia University,
      \href{mailto:melina.mailhot@concordia.ca}{\nolinkurl{melina.mailhot@concordia.ca}}.},
    Nahid Sadr\footnote{Département de Mathématiques, Université de Sherbrooke, \href{mailto:nahid.sadr@usherbrooke.ca}{\nolinkurl{nahid.sadr@usherbrooke.ca}}.}
    \par}
\end{center}
\par\smallskip
\begin{abstract}
  A stochastic representation and sampling algorithm for Morillas-type
  copula-to-copula transformations and related distortions of multivariate
  distribution functions is derived, resulting as a byproduct in a novel
  sampling scheme for Archimedean and Archimax copulas.  This closes a
  methodological gap and facilitates simulation-based applications of distorted
  copulas.  For stable tail dependence functions (stdfs), a Morillas-type
  distortion framework is introduced, where monomial distortions with exponents
  below $1$ are shown to preserve stdfs via a domain-restricted Pexider equation
  analysis. This characterization is leveraged to identify distortions
  preserving extreme value copulas, and convex combinations of distorted stdfs
  are proposed to increase flexibility in extremal dependence modeling. The
  impact of distortions on maximum domain of attraction limits is also
  analyzed. Explicit limiting EVC distortions are identified under
  non-restrictive regular variation assumptions.  Examples of absolutely
  monotone distortions allowing to fine-tune the extreme value behavior after
  distortion, but also of non-regularly varying $2$-absolutely monotone
  distortions are given.
\end{abstract}
\minisec{Keywords}
Absolute monotonicity, Archimax copula, Archimedean copula, maximum domain of attraction,
Morillas-type distortions, multivariate extreme value theory, Pexider functional equation, stable tail dependence functions, stochastic representation.
\section{Introduction}
Copula-to-copula transformations have appeared in various forms throughout the
literature, see for example \cite{Khoudraji1995}, \cite{Morillas2005}, or
\cite{liebscher2008} for early developments. Later studies have generalized or
unified these ideas; for instance, \cite{fischer2012constructing} proposed a
framework that extends both Liebscher's Archimedean family extension and
Morillas transformations. \cite{durante2009construction} introduced a more
general method for building non-exchangeable bivariate distribution functions,
which includes Khoudraji transforms as a special case.  Similarly,
\cite{kolesarova2013new} developed an approach for constructing new copulas from
existing ones through a dual representation. In a related line of research,
\cite{cuadras2009constructing} introduced copulas based on weighted geometric
means of given dependence structures, and \cite{HerrmannHofertSadr2023} proposed
index-mixed constructions, combining several copulas via random index
vectors. More recently, uniformity preserving transformations from
$[0,1]$ to $[0,1]$ such as the W-transform of \cite{hofert2025w}, with a special
case being the V-transforms of \cite{McNeil2021modelling}, are constructed to maintain marginal
uniformity while being able to alter the dependence structure.

Within this body of work, Morillas-type distortions, the focus of this paper,
provide a copula-to-copula mapping that is analytically simple and sufficiently
versatile to generate interesting dependence structures; see
\cite{durrleman2000simple,GenestRivest2001,KlementMesiarPap2005,durante2005copula,Morillas2005}
for the initial definition and properties, \cite{valdez2011distortion} for
further details, and \cite{crane2008using} for an application to risk pricing.
For a given copula $\Copula$ the Morillas-type distorted copula $\Copula_{\pD}$ takes the form
\begin{align}
  \Copula_{\pD}(u_1,\dots,u_d) := \pD\bigl(\Copula(\qD(u_1),\dots,\qD(u_d))\bigr)\label{eq_MorillasDef}
\end{align}
for a suitable non-decreasing function $\pD:[0,1]\to[0,1]$.  Well known special
cases of this construction are for example Archimedean, see \cite{McNeilNeslehova2009}, and Archimax copulas, see \cite{CharpentierFougeresGenestNeslehova2014}.

Despite the simplicity of Morillas-type transformations, a stochastic
representation or simulation algorithm for $\Copula_{\pD}$, the analysis of the
extremal behavior of $\Copula_{\pD}$, and related relevant properties of the
functions $\pD$ have remained unexplored in the literature. This gap is
particularly relevant in the context of multivariate extreme value theory, where
typically only specific dependence structures arise when considering limits of
suitably stabilized componentwise maxima.  Therefore, it is natural to examine
how this copula-to-copula transformation impacts limiting dependence structures
linked to extremes.

In extreme value theory, a classical result identifies multivariate extreme
value distributions as the only possible limits for affine stabilizations of
componentwise block maxima of an iid sequence of random vectors $(\bfX_i)$,
where $\bfX_i = (X_{i,1},\dots,X_{i,d})$ has continuous margins and unique
copula $C$; see \cite[Chapter~5]{Resnick1987}. More precisely, let
$\bfM_n = (M_{n,1},\dots,M_{n,d})$ denote the vector of block maxima, where
$M_{n,j} = \max(X_{1,j},\dots,X_{n,j})$ is the maximum over the $j$th component sample,
$j=1,\dots,d$. If, for a non-degenerate limiting distribution
$\EVDist$,
\begin{align}
  \lim_{n\to\infty} \Prob\left((\bfM_n - \bbn_n)/\ban_n\leq \bfx\right) = \EVDist(\bfx), \quad \bfx=(x_1,\dots,x_d)\in\IR^d, \label{eq_multivariateGEVLimit}
\end{align}
where $\bbn_n = (\bn_{n,1},\dots,\bn_{n,d}) \in \IR^d$ and
$\ban_n = (\an_{n,1},\dots,\an_{n,d}) \in (0,\infty)^d$ are sequences of
stabilizing vectors, then $\EVDist$ necessarily belongs to the class $\GEV_d$ of
$d$-dimensional extreme value distributions, that is $\EVDist\in\GEV_d$;
all operations in~\eqref{eq_multivariateGEVLimit} are understood to be componentwise.
Elements of $\GEV_d$ can be characterized in several ways, for example using extreme value copulas, spectral
measures, exponent measures, tail copulas, Pickands dependence functions or stable tail dependence
functions. For the purpose of this article, we focus on representations in
terms of extreme value copulas and stable tail dependence functions.

The main contributions of our article are:
\begin{enumerate}
\item We provide a stochastic representation and sampling scheme for $\Copula_\pD$ for absolutely monotone $\pD$;
\item In analogy to the Morillas transform for copulas, we propose and analyze a
  Morillas-type transform for stable tail dependence functions and identify the admissible distortions in this case; and
\item We characterize the maximum domain of attraction of the distorted
  copula $\Copula_\pD$.
\end{enumerate}
The remainder of this article is structured as follows.
Section~\ref{sec_cdf_distortions} introduces Morillas-type transformations,
addresses their main properties and establishes a stochastic representation and
sampling algorithm for $\Copula_\pD$.  Section~\ref{sec_basic_notions} details
the connection between multivariate generalized extreme value distributions,
stable tail dependence functions and extreme value copulas, and introduces a
Morillas-type distortion for stable tail dependence
functions. Section~\ref{sec_MDA_for_C_D} presents the maximum domain of
attraction for the distorted copula $\Copula_\pD$ and related results. Finally,
Section~\ref{sec_conclusion} concludes with a summary of our results.  Proofs of
all results are provided in the Supplementary Material.

\section{Distortions of joint distribution functions, copulas and stable tail dependence functions}\label{sec_cdf_distortions}
\subsection{Properties}\label{sec_properties_distortions}
A \emph{distortion function} $\pD \colon [0,1] \to [0,1]$ is a distribution function with $\pD(0)=0$ and $\pD(1)=1$.
In this section we review and consolidate results concerning distortions of a $d$-variate joint distribution function $\JointDist$ via
\begin{align}
  \bfx \mapsto \pD\circ \JointDist(\bfx), \quad \bfx\in\IR^d,\label{eq_F_D}
\end{align}
and ask when $\pD\circ \JointDist$ is again a valid joint distribution function.
Aside from a succinct presentation of earlier results, the main contribution of
this section is a stochastic representation for $\pD\circ \pF$ and its associated
copula $\Copula_{\pD}$, which seems to be missing in the literature so far.

When considering $\pD\circ\JointDist$, the distortion simultaneously impacts the
margins and the underlying dependence structure of $\JointDist$.  A direct
calculation shows that that the $j$th margin $F_j$ of $\JointDist$ is given by
$\pD\circ F_j$, $j=1,\dots,d$. As a distribution function, $\pD$ is
right-continuous and non-decreasing, and so is $\pD\circ F_j$, $j=1,\dots,d$.
However, every $d$-dimensional joint distribution function has to be
$d$-increasing, see \cite[Theorem~1.2.13, p.~7]{DuranteSempi2016}. As such,
$\pD$ critically determines whether $\pD\circ\JointDist$ is a valid joint
distribution function.

Similar to the distortion in~\eqref{eq_F_D}, \cite{durrleman2000simple}, \cite{GenestRivest2001}, \cite{KlementMesiarPap2005}, \cite{durante2005copula}, and, in full generality, \cite{Morillas2005} consider for a distortion function $\pD$
the construction~\eqref{eq_MorillasDef}, where the pre-composition with the
inverse function $\qD$ is necessary to guarantee that the margins of
$\Copula_{\pD}$ remain uniform. This construction has interesting properties,
for example for two distortions $\pD_1$ and $\pD_2$ we have
$(\Copula_{\pD_2})_{\pD_1} = \Copula_{\pD_1\circ\pD_2}$ so that distortions of
distorted copulas are distorted copulas.  Furthermore, Morillas-type transforms
also preserve exchangeability, %
if $\Copula(u_{\sigma(1)},...,u_{\sigma(d)}) = \Copula(u_1,\dots,u_d)$ for any
permutation $\sigma\colon\{1,\dots,d\}\to\{1,\dots,d\}$, then a direct
calculation shows that the same property also holds for $\Copula_{\pD}$.

Next, we consider the crucial question concerning the conditions under which the
copula $\Copula_{\pD}$ in~\eqref{eq_MorillasDef} is a copula for \emph{any}
copula $C$. For $f \colon [a,b] \to \IR$, $a<b$, $k\in\IN_0$ and $h>0$ such that
$x,x+kh\in[a,b]$, the \emph{$k$th order forward difference with step size $h$}
of $f$ is
$(\Delta_h^k f)(x) := (\Delta_h^{k-1} f)(x+h) - (\Delta_h^{k-1} f)(x) =
\sum_{j=0}^{k}\binom{k}{j}(-1)^{k-j}f(x+jh)$, where $(\Delta_h^0 f)(x) :=f(x)$.
The closely related \emph{$k$th order divided difference with step size $h$} is
$D_h^k f := h^{-k} \Delta_h^k f$ whenever $x,x+kh\in [a,b]$. Note that
$(\Delta_h^k f)(x) \geq 0 \Leftrightarrow (D_h^k f)(x) \geq 0$.  For $d\in\IN_0$,
$f\colon[a,b]\to\IR$ is \emph{$d$-absolutely monotone} if
$(\Delta_h^k f)(x) \geq 0$ for all $k \in \{0,1,\dots,d\}$, all $x\in[a,b]$ and
$h>0$ such that $x+kh\in[a,b]$. And if $f$ is $d$-absolutely monotone for all
$d\in\IN_0$, $f$ is \emph{absolutely monotone}.
Such mild restrictions on $\pD$ leave a wide range
of possible transformations $\pD$ of \emph{Morillas-type}. For example, any
$\pD(u)=\sum_{n=1}^\infty p_n u^n$, $u\in[0,1)$, for a probability mass function
$(p_n)_{n\in\IN}$, that is if $\pD$ is a probability generating function where $p_0=0$. In
Section~\ref{sec:stoch:rep} we utilize this connection for establishing
a stochastic representation and thus sampling approach for
$\Copula_{\pD}$ and $\pD\circ\JointDist$ when $\pD$ is absolutely monotone.

Let us again consider the $d$-absolutely monotone case $f:[a,b]\to\IR$. As an
example, a function is $2$-absolutely monotone if and only if it is
non-negative, non-decreasing and convex; see \cite[p.~58]{Ressel2012}. As a
convex function $f$ is then also continuous on $(a,b)$, and it can be continuously
extended to the left endpoint $a$ of its domain. The fact that any
$d$-absolutely monotone function for a $d\geq 2$ is continuous on $(a,b)$ allows
one to apply the theorems of \cite{Popoviciu1934,BoasWidder1940}; see
\cite[Theorem~18.10, p.~118 and 23 Appendix A, p.~137 ff.]{Khare2022} for a
modern treatment. As a result and as noted by \cite[p.~58]{Ressel2012}, for
$d\geq 2$, a $d$-absolutely monotone function $f$ has existing derivatives
$f', f'', \dots, f^{(d-2)}$ on $(a,b)$, all of which are continuous and
continuously extendable at $a$, non-negative, increasing and convex, and the
function $f^{(d-2)} \colon (a,b) \to \IR$ has non-decreasing left- and
right-hand side derivatives. Consequently, an absolutely monotone distortion $\pD$ has
non-negative derivatives of any non-negative integer order.

It is often easier, although less general, to verify that a non-negative function has
non-negative derivatives $f',f'',\dots,f^{(d)}$ instead of having to verify that it is
$d$-absolutely monotone by definition; see \cite[p.~58]{Ressel2012}.  Similarly,
if $f^{(k)}\geq 0$ for all $k \geq 0$, then
$f$ is absolutely monotone.  The following
Theorem~\ref{thm_distortion_master_theorem} is the basis for distortions of
copulas and joint distribution functions.
Part~\ref{thm_distortion_master_theorem_copula} was in full generality first
established by \cite[Theorem~4.7, p.~182]{Morillas2005} and \cite[Theorem~3.6, p.~99]{morillas2005b}, while bivariate
versions were also published in \cite[Theorem~1, p.~2]{durrleman2000simple},
\cite[Theorem~2.4, p.~427]{KlementMesiarPap2005} and \cite[Theorem~3.1,
p.~650]{durante2005copula}. Parts~\ref{thm_distortion_master_theorem_cdf} and
\ref{thm_distortion_master_theorem_possible_copula} are taken from
\cite[Theorem~2, p.~59; Theorem~5, p.~61]{Ressel2012}, where the bivariate case
was also already addressed in \cite[Remark~3.5, p.~652]{durante2005copula}.  As
mentioned, a $d$-absolutely monotone function is always continuous on its domain
except possibly at its right endpoint. As such the continuity of $\pD$ at $1$
plays a non-trivial role in the following theorem. For a distortion function
$\pD$, the condition of $\pD$ being a $d$-absolutely monotone surjection is
equivalent to $\pD$ being $d$-absolutely monotone and continuous at $1$.  Let
$\AMDCO{d}$ denote the set of $d$-absolutely monotone distortions continuous at
$1$, with the obvious chain of inclusions $\AMDCO{2} \supsetneq \AMDCO{3} \supsetneq \ldots \supsetneq \AMDCO{\infty}$,
where $\AMDCO{\infty}$ is the set of absolutely monotone distortions continuous
at $1$.

\begin{theorem}[Distortion of copulas and joint distribution functions; \cite{Morillas2005}, \cite{Ressel2012}]\label{thm_distortion_master_theorem}
For $\pD\in\AMDCO{d}$, $d\geq 2$, we have the following results: %
\begin{enumerate}
    \item\label{thm_distortion_master_theorem_copula} If $\Copula$ is a $d$-copula, then $\Copula_{\pD}$ defined in~\eqref{eq_MorillasDef} is a $d$-copula.
    \item\label{thm_distortion_master_theorem_cdf} If $F$ is a $d$-dimensional distribution function with margins $F_1,\dots,F_d$, then $\pD \circ F$ defined in~\eqref{eq_F_D} is a $d$-dimensional distribution function with margins $\pD\circ F_j$, $j\in\{1,\dots,d\}$.
    \item\label{thm_distortion_master_theorem_possible_copula} If $\Copula$ is a copula of the $d$-dimensional distribution function $F$, then $\Copula_{\pD}$ is a copula of $\pD\circ F$.
\end{enumerate}
\end{theorem}
\begin{example}[Necessity of left-continuity at $1$, $\AMDCO{d} \supsetneq \AMDCO{d+1}$ and sufficiency of Theorem~\ref{thm_distortion_master_theorem}]
  \begin{enumerate}
  \item To see that the left-continuity of $\pD$ at $1$ is indeed necessary,
    consider, for any $c\in(0,1]$, the function $\pD_c(u) = cu$ for $u\in[0,1)$
    and $\pD_c(1)=1$. While $\pD_c$ is a distortion function for any
    $c\in(0,1]$, it is continuous at $1$ if and only if $c=1$.  The generalized
    inverse of $\pD_c$ is given by $\qD_c(y) = 1$ for $y\in[c,1]$ and by
    $\qD_c(y) = y/c$ for $y\in[0,c)$.  As such, we have
    $\pD_c \circ \qD_c (y) = 1$ for $y\in[c,1]$ and $\pD_c \circ \qD_c (y) = y$
    for $y\in[0,c)$. Hence for $c\in(0,1)$, the distortion $\Copula_{\pD_c}$
    does not have uniform margins and is therefore not a copula.
\item To see that $\AMDCO{d} \supsetneq \AMDCO{d+1}$ consider, for $d\in\IN$,
  the distortion function $\pD(u) = u^{d-1/2}$ for $u\in[0,1]$. One can directly
  check that $\pD^{(d-2)}$ is increasing and convex, and that $\pD^{(d-1)}$ is
  increasing.  Therefore $\pD$ is $d$-absolutely monotone.  However,
  $\pD^{(d-1)}$ is concave and hence $\pD$ is not $(d+1)$-absolutely monotone
  (nor absolutely monotone of any higher order).
\item Theorem~\ref{thm_distortion_master_theorem} only addresses sufficiency.
  To see this, consider the comonotone copula $M(u,v)=\min(u,v)$, $u,v\in[0,1]$,
  and, for example, $\pD(u) = \sqrt{u}$.  As a strictly concave function, $\pD$
  cannot be $2$-absolutely monotone. However, since $\pD$ is strictly increasing
  and continuous, we can interchange the pre- and post-composition with the
  minimum to see that $M_{\pD} = M$.  The specific choice of $\pD$ is of course
  not relevant, and any continuous and strictly increasing function which is not
  $2$-absolutely monotone would have been sufficient. The importance of
  Theorem~\ref{thm_distortion_master_theorem} lies in the fact that it allows
  one to transform \emph{any} copula into a valid copula, that is it answers
  which distortions act faithfully on the biggest possible domain of all
  copulas, whereas this example here only works for a specific combination of
  copula (the comonotone one) and distortion function (the square root). The
  somewhat analogous question about which largest set of distortions
  $\Theta(\Copula)$, $\AMDCO{d}\subseteq\Theta(\Copula)$, faithfully transforms
  a given \emph{fixed} $d$-copula $\Copula$ into another $d$-copula is
  investigated in \cite{GhiselliRicci2024}.
  \end{enumerate}
\end{example}

A first natural question concerns the impact of a distortion.  For rectangles,
the following result provides an answer.  While this result in principle allows
to construct bounds for sets which are well approximated by rectangles, for
example Jordan measurable sets, bounds for the total variation distance, that is
when considering general Borel-measurable sets instead of rectangles, do not
follow from the same argument.
The following Proposition~\ref{prop_distorted_measure_difference} specifically
applies to any strictly increasing $\pD,\pD'\in\AMDCO{d}$, but is formulated
under slightly weaker assumptions which are necessary for its application in the proof of
Lemma~\ref{rate_function_convergence}.
\begin{proposition}[Maximal difference of distorted rectangle volumes]\label{prop_distorted_measure_difference}
Let $\Copula$ be a $d$-copula and $\pD'$ a distortion function such that $\Copula_{\pD'}$ is a valid $d$-copula.
\begin{enumerate}
\item\label{prop_max_difference_I}
If $\pD'$ and $\pD$ are continuous and strictly increasing distortion functions such that $\Vert \pD' - \pD \Vert_\infty = \sup_{u \in [0,1]} \vert \pD'(u) - \pD(u) \vert \leq \varepsilon$, $\varepsilon > 0$, then
\begin{align*}
  \sup_{\bfu \in [0,1]^d} \abs{\Copula_{\pD'}(\bfu) - \Copula_{\pD}(\bfu)} \leq (d+1)\varepsilon,
\end{align*}
where this bound holds even if, depending on $\pD$, $\Copula_{\pD}$ as defined in \eqref{eq_MorillasDef} is \emph{not} a $d$-copula.
Specifically, for $\pD' = \id$ we have $\Vert \Copula-\Copula_{\pD}\Vert_{\infty} \leq (d+1)\varepsilon$.
\item
  The corresponding $\Copula_{\pD'}$-volume $\vol_{(\bfa, \bfb]}\Copula_{\pD'}$ and $\Copula_{\pD}$-volume $\vol_{(\bfa, \bfb]}\Copula_{\pD}$ (with $\vol_{(\bfa, \bfb]}\Copula_{\cdot} = \sum_{\bm{i}\in \{0,1\}^d}$ $(-1)^{\sum_{j=1}^d i_j} \Copula_{\cdot}(\bm{u}^{(\bm{i})})$, where $\bfu^{(\bm{i})} = (a_1^{i_1}b_1^{1-i_1},\dots,a_d^{i_d}b_d^{1-i_d})$, $\bm{i}\in\{0,1\}^d$) satisfy
\begin{align*}
  \sup_{(\bfa,\bfb]\subseteq[0,1]^d} \abs{\vol_{(\bfa,\bfb]}\Copula_{\pD'} - \vol_{(\bfa,\bfb]}\Copula_{\pD}} \leq 2^d (d+1) \varepsilon,
\end{align*}
where $\bm{0}\le\bfa \leq \bfb\le\bm{1}$ componentwise and $(\bfa, \bfb]= \prod_{i=1}^d(a_i,b_i]$.
Specifically, for $\pD' = \id$ we have
\begin{align*}
  \sup_{(\bfa,\bfb]\subseteq[0,1]^d} \abs{\vol_{(\bfa,\bfb]}\Copula - \vol_{(\bfa,\bfb]}\Copula_{\pD}} \leq 2^d (d+1) \varepsilon.
\end{align*}
\end{enumerate}
\end{proposition}
Given that copulas are distribution functions, the $\Copula$-volume of a copula $\Copula$ is in fact equal to the associated copula measure of the rectangle under consideration; in Definition~\ref{def_stdf} we will later also see volumes of other functions than copulas.

As a direct consequence of Proposition~\ref{prop_distorted_measure_difference}, if a sequence $(\pD_n)$ of distortion functions converges uniformly to a distortion $\pD$, then the respective sequence of distorted copulas
converges uniformly.
\begin{corollary}[Uniform convergence of $\Copula_{\pD_n}$ to $\Copula_{\pD}$]\label{cor_convergence_C_Dn_to_C_D}
  Let $\Copula$ be a $d$-copula and $(\pD_n)$ a sequence of strictly increasing
  and continuous distortion functions.  If $\pD_n$ converges pointwise to a
  strictly increasing and continuous distortion function $\pD$, and
  $\Copula_{\pD}$ is a $d$-copula, then $(\Copula_{\pD_n})$ converges uniformly
  to $\Copula_{\pD}$, that is
  $\lim_{n\to\infty}\Vert \Copula_{\pD_n}-\Copula_{\pD} \Vert_{\infty} = 0$.
  This convergence holds even if $\Copula_{\pD_n}$ fails to be a $d$-copula for
  some or all $n\in\IN$.
\end{corollary}
An example that highlights the importance of the case in which
$\Copula_{\pD_n}$ is not a copula in
Proposition~\ref{prop_distorted_measure_difference}
and Corollary~\ref{cor_convergence_C_Dn_to_C_D}
is, for example, the sequence
of distortions $\pD_n(u) = u^{a_n}$, where $(a_n)$ is a non-negative sequence of
real numbers such that $a_n\to 1$.  For an arbitrary $d$-copula $\Copula$ there
is in general no guarantee that $\Copula_{\pD_n}$ is itself a $d$-copula since
$\pD_n$ may fail to be $d$-absolutely monotone depending on the sequence
$(a_n)$.  However, we still have the (intuitively expected) uniform convergence
$\Copula_{\pD_n} \to \Copula$ via Corollary~\ref{cor_convergence_C_Dn_to_C_D}.
The application of Proposition~\ref{prop_distorted_measure_difference} in the proof of
Lemma~\ref{rate_function_convergence} later follows a similar logic.

In light of Theorem~\ref{thm_distortion_master_theorem}, the subclass
$\AMDCO{\infty}$ of absolutely monotone distortions is of special interest since
they can be used to distort distribution functions and copulas in any
dimension. Additionally, by \cite[Theorem~2, p.~223]{feller1971}, see also
\cite[Lemma~2.1.6, p.~51]{Hofert2010},
absolutely monotone functions $f$ on $[0,1)$ have a power series representation
$f(u) = \sum_{n=0}^{\infty} p_n u^n$ for $(p_n)_{n\in\IN_0}\ge 0$, naturally
connecting those with $f(1)=1$ to probability generating functions; this is
reminiscent of the connection between completely monotone functions and
Laplace--Stieltjes transforms by Bernstein's theorem, see \cite[Theorem~1,
p.~439]{feller1971}. The following result revisits \cite[Theorem~2,
p.~223]{feller1971} in the context of distortion functions, which is warranted
since, as stated in \cite[p.~223]{feller1971}, the function $f(x) = (1-x)^{-1}$
has a power series representation with non-negative coefficients, yet fails to
be bounded at $1$ and is thus not a distortion function. Similarly, one can have
a power series $f$ that is continuous at $1$ but with a value $f(1)\neq 1$,
again preventing $f$ from being a distortion function. The following result also
directly implies that absolutely monotone distortions that are continuous at $1$
are strictly increasing and thus bijective.
\begin{proposition}[{Probability generating function representation for $\AMDCO{\infty}$}]\label{prop_pgf_representation}
  $\pD\in\AMDCO{\infty}$ if and only if there exists a discrete random
  variable $N$ on $\IN$ such that $\pD(u) = \sum_{n=1}^{\infty} \Prob(N=n)u^n$, $u\in[0,1]$.
  This random variable $N$ is called the \emph{discrete generator} of $\pD$.
\end{proposition}

The following corollary allows one to identify when the density of an absolutely monotone distortion is bounded.
\begin{corollary}[Bounded densities of absolutely monotone distortions]\label{cor_bounded_density}
  Let $\pD$ be a distortion function such that
  $\pD(u) = \sum_{n=1}^{\infty} \Prob(N=n)u^n$ for all $u\in[0,1]$, where $N$ is
  a discrete random variable on $\IN$.  Then the density
  $\dD = \frac{\d}{\d u}\pD$ is bounded on $[0,1]$ if and only if
  $\E[N]<\infty$.
\end{corollary}

\subsection{Stochastic representation}\label{sec:stoch:rep}
\cite{Morillas2005} did not provide a stochastic representation for
$\Copula_{\pD}$ in~\eqref{eq_MorillasDef}.  The same is true for the bivariate
investigations in
\cite{durrleman2000simple,GenestRivest2001,KlementMesiarPap2005,durante2005copula}.
This section provides such a representation and related sampling algorithm for
$\pD\in\AMDCO{\infty}$, where the utilization of the probability generating function is reminiscent of the utilization of the Laplace transform in the stochastic representation and related sampling algorithm of \cite{marshallolkin1988} for Archimedean
copulas.
\begin{proposition}[Stochastic representation for Morillas transformed copulas $\Copula_{\pD}$ with $\pD\in\AMDCO{\infty}$]\label{prop:stoch:rep}
  Let $\pD\in\AMDCO{\infty}$ with representation
  $\pD(u)=\sum_{n=1}^\infty p_n u^n$, $u\in[0,1]$, where $p_n\ge 0$,
  $\sum_{n=1}^\infty p_n=1$. For $d\ge 2$, let $\Copula$ be a $d$-copula, and
  let $N$ be a discrete random variable with probability mass function
  $\Prob(N=n)=p_n$, $n\in\IN$.  If $\bm{W}_N=(W_{N,1},\dots,W_{N,d})$ is a $d$-dimensional
  random vector such that $\bm{W}_N\mid N=n\sim\Copula^n$, $n\in\IN$, we have
  \begin{align}
    \bm{U}_{\pD}:=(\pD(W_{N,1}),\dots,\pD(W_{N,d}))\sim \Copula_{\pD}.\label{eq_MorillasDef_stochrep}
  \end{align}
  One possible stochastic representation for $\bm{W}_N\mid N=n$ is
  \begin{align}
    \bm{W}_n=(\max_{1\le i\le n}\{U_{i,1}\},\dots,\max_{1\le i\le n}\{U_{i,d}\}),\label{eq:sample:bottleneck} %
  \end{align}
  where $\bm{U}_1,\dots,\bm{U}_n\isim\Copula$ for $\bfU_i = (U_{i,1},\dots,U_{i,d})$, $i\in\{1,\dots,n\}$.
\end{proposition}
Proposition~\ref{prop:stoch:rep} immediately suggests the following algorithm to
draw a sample from $\Copula_{\pD}$.  To utilize this sampling scheme, it needs
to be possible (i) to sample from the discrete generator $N$ associated to
$\pD$, and (ii), at least, to sample from the initial copula $\Copula$.
\begin{algorithm}[Sampling Morillas transformed copulas $\Copula_{\pD}$ for $\pD\in\AMDCO{\infty}$]\label{algo_sampling_Morillas}
  With the setup as in Proposition~\ref{prop:stoch:rep}, to
  sample~\eqref{eq_MorillasDef_stochrep}, do:
  \begin{enumerate}
  \item Sample $N$ with probability mass function $\Prob(N=n)=p_n$, $n\in\IN$.
  \item Sample $\bm{W}_N=(W_{N,1},\dots,W_{N,d})\sim\Copula^N$, %
    for example using the stochastic representation in \eqref{eq:sample:bottleneck}.
  \item Return $\bm{U}_{\pD}=(\pD(W_{N,1}),\dots,\pD(W_{N,d}))\sim \Copula_{\pD}$.
  \end{enumerate}
\end{algorithm}
If $\pD\in\AMDCO{\infty}$, Algorithm~\ref{algo_sampling_Morillas} can also be used for sampling from
$\pD\circ\jointF$ for any $d$-variate joint distribution function $\jointF$ with
margins $F_1,\dots,F_d$ and copula $\Copula$. According to
Theorem~\ref{thm_distortion_master_theorem}~\ref{thm_distortion_master_theorem_cdf} and
\ref{thm_distortion_master_theorem_possible_copula}, $\pD\circ\jointF$ has
margins $\pD\circ F_j$, $j\in\{1,\dots,d\}$, and (possibly non-unique) copula
$\Copula_{\pD}$. Furthermore, since $\pD$ is bijective by assumption, see
Proposition~\ref{prop_pgf_representation}, a direct calcuation shows that the
$j$th marginal quantile function is $F_j^{-1}\circ\qD$, $j\in\{1,\dots,d\}$. As
such, once a sample $\bm{U}_{\pD}\sim\Copula_{\pD}$ is drawn according to
Proposition~\ref{prop:stoch:rep}, one knows that
$\bigl(F_1^{-1}(\qD(U_{\pD,1})),\dots,F_d^{-1}(\qD(U_{\pD,d}))\bigr)\sim \pD\circ\jointF$.

The advantage of the stochastic representation of
Proposition~\ref{prop:stoch:rep} and Algorithm~\ref{algo_sampling_Morillas} lies
in their generality in that they apply to any copula $\Copula$ that can itself
be sampled.  If the discrete generator $N$ associated to $\pD$ puts mass
on large integers, for example if $\E[N]=\infty$, then this generality can come
at the price of an increased run time. However, not necessarily. If we can
directly, that is without the maxima as in~\eqref{eq:sample:bottleneck}, sample
$\bm{W}_N$ from $\Copula^N$, then the stochastic representation remains
efficient for sampling.
Distortions of the independence copula $\Pi$ provide a prime example of this circumstance in which case we have $\Pi^N(\bfu) = \Pi(\bfu^N)$
so that sampling $\bm{W}_N$ simply results in returning
$(U_1^{1/N},\dots,U_d^{1/N})$, where $U_1,\dots,U_d\isim\U(0,1)$.  A similar
argument also applies to extreme value copulas in place of $\Pi$, see
Definition~\ref{def_extreme_value_copula} later. Based on this observation we
provide a novel efficient sampling scheme for a subset of Archimedean and
Archimax copulas in Remark~\ref{remark_Archimedean_sampling}.  While this
constitutes a large class of copulas to efficiently apply
Algorithm~\ref{algo_sampling_Morillas} to, the following example
provides a copula outside this class where efficient sampling is still possible.
\begin{example}[Efficient sampling of distorted Clayton copulas]\label{example_sampling_distorted_Clayton}
  An example of a non-extreme value copula where efficient sampling according to
  Algorithm~\ref{algo_sampling_Morillas} is possible is given by the Clayton
  copula $\ClCopula_{\theta}$ with parameter $\theta > 0$; see
  \cite[Section~4.2, p.~114~ff.]{Nelsen2006}. Here one can directly calculate
  that $(\ClCopula_{\theta}(\bfu))^n = \ClCopula_{\theta/n}(\bfu^n)$.  As such,
  one can use the efficient Marshall--Olkin sampling algorithm for Clayton
  copulas to directly sample $\bm{W}_N \sim (\ClCopula_{\theta})^N$ without
  resorting to the maxima-based alternative in~\eqref{eq:sample:bottleneck}.  We
  employ this strategy to obtain Figure~\ref{fig_Clayton_zeta}, where we show
  samples of size $n=10\,000$ from an undistorted and from a distorted Clayton
  copula with parameter $\theta=\pi$. Here we consider the absolutely monotone
  zeta distortion $\pD_s$ with parameter $s=1.5$ discussed in detail in
  Example~\ref{example_zeta_distortion} later.  Most
  notably, the figure shows that the upper tail independence of the Clayton
  copula is transformed into upper tail dependence by the distortion.  This
  behavior is also theoretically investigated and confirmed in
  Section~\ref{sec_MDA_for_C_D}, specifically Corollary~\ref{cor_Pi_to_Gumbel}
  and the discussion thereafter, where we find that, for the chosen parameter
  values, $(\ClCopula_{\theta})_{\pD_s}$ is in the maximum domain of attraction
  of a Gumbel--Hougaard copula with a parameter value of $2$.
  \begin{figure}
    \begin{center}
      \includegraphics[width=0.48\textwidth]{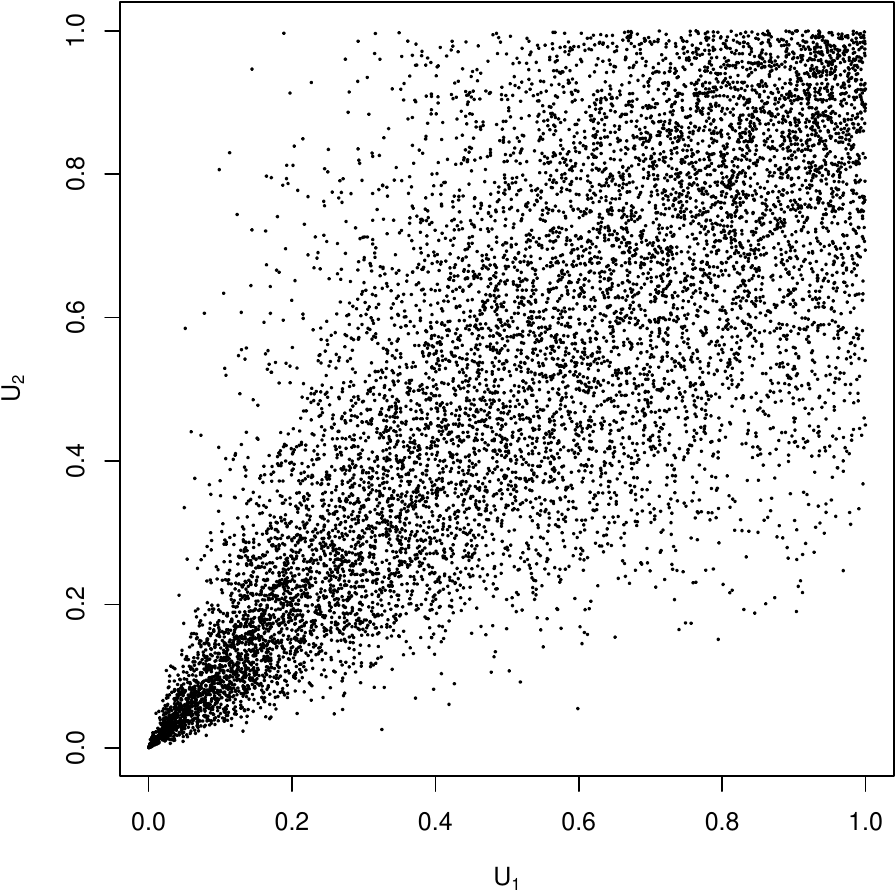}\hfill
      \includegraphics[width=0.48\textwidth]{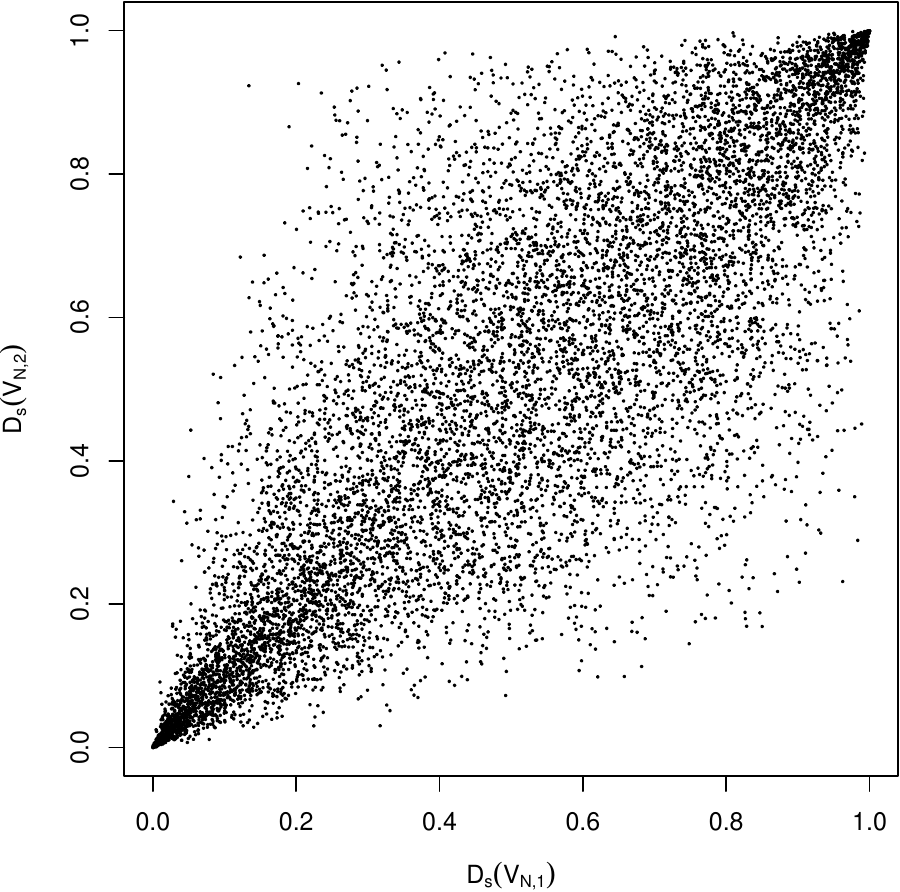}
    \end{center}
    \caption{Left: Sample of size $n=10\,000$ of a Clayton copula
      $\ClCopula_{\theta}$ with parameter $\theta=\pi$. Right: Sample of size
      $n=10\,000$ of $(\ClCopula_{\theta})_{\pD_s}$ where $\pD_s$ is a zeta
      distortion with parameter $s=1.5$; see
      Example~\ref{example_zeta_distortion}.}
    \label{fig_Clayton_zeta}
  \end{figure}
\end{example}

\section{Multivariate extreme value distributions and Morillas-type distortions of stable tail dependence functions}\label{sec_basic_notions}
In this section we first briefly review necessary facts about multivariate
extreme value distributions, extreme value copulas and stable tail dependence
functions. As mentioned in the introduction, our main goal is then to define a
Morillas-type distortion for the class of stabe tail dependence functions and
characterize the admissible distortions in this case.

\subsection{Multivariate max-stable distributions and their connections to extreme value copulas and stable tail dependence functions}
The following definition is from \cite[Equation~(5.17), p.~264]{Resnick1987}.
\begin{definition}[Max-stable distributions]
  Let $d\in\IN$ and $\jointF$ a $d$-dimensional distribution function.
  Then $\jointF$ is \emph{max-stable} if and only if for every $t>0$ there exist
  $a_j(t)>0$ and $b_j(t)\in\IR$, $j\in\{1,\dots,d\}$, such that
  \begin{align}\label{eq_max_stable}
    \jointF^t(\bfx) = \jointF(a_1(t)x_1 + b_1(t),\ldots,a_d(t)x_d + b_d(t)),\quad \bfx=(x_1,\dots,x_d)\in\IR^d.
  \end{align}
  The class of all $d$-dimensional max-stable distributions is
  denoted by $\GEV_d$.%
\end{definition}
In case $d=1$ the class of max-stable distributions is described precisely via
$H\in\GEV_1$ if and only if $H$ is either a location-scale transform of the
Gumbel, (extremal) Weibull or the Fréchet distribution; see the Fisher--Tippett--Gnedenko theorem
in \cite[Theorem~3.2.3]{embrechtsklueppelbergmikosch1997}. %

The multivariate setting complicates matters since additional to the constraint
that each margin has to be max-stable, \eqref{eq_max_stable} also places a
restriction on the underlying dependence structure.  Taking the point-of-view of
Sklar's theorem one can argue, see for example
\cite[p.~129]{GudendorfSegers2010}, that $\EVDist\in\GEV_d$ if and only if all
margins are max-stable and the associated copula belongs to the class
of extreme value copulas.
\begin{definition}[Extreme value copula]\label{def_extreme_value_copula}
  A copula $\Copula^*$ is called \emph{extreme value copula (EVC)} if there exists a copula $\Copula$ such that
  $\Copula(u_1^{1 / n}, \ldots, u_d^{1 / n})^n \to \Copula^*(u_1, \ldots, u_d)$ for $n\to\infty$
  and all $(u_1,\ldots,u_d) \in [0,1]^d$. If so, the copula $\Copula$ is said to
  be in the \emph{maximum domain of attraction (MDA)} of $\Copula^*$ and we write
  $\Copula\in\MDA(\Copula^*)$.  We denote the set of $d$-dimensional EVCs by $\SetEVTCopula_d$.
\end{definition}
As shown for example in \cite[Theorem~3.3.5]{Nelsen2006} or
\cite{GudendorfSegers2010}, we have $\Copula^*\in\SetEVTCopula_d$ if and only if
$\Copula^*(u_1^t,\dots,u_d^t)^{1/t} = \Copula^*(u_1,\dots,u_d)$ for all $t>0$
and $(u_1,\dots,u_d)\in[0,1]^d$, where copulas with the latter property are also
called \emph{max-stable copulas}.

Describing EVCs, or, equivalently max-stable copulas, can be done in several
ways, via stable tail dependence functions (see
\cite{Huang1992,DreesHuang1998,Ressel2013}), Pickand's dependence functions (see
\cite{Pickands1981}), tail copulas (see \cite{SchmidtStadtmuller2006}) or
spectral measures (see \cite{deHaanResnick1977}). We focus on stable tail
dependence functions; for the following characterization that we take as
definition, see \cite[Theorem~6, p.~255]{Ressel2013}.
\begin{definition}[Stable tail dependence function (stdf)]\label{def_stdf}
  A function $\STDF \colon [0,\infty)^d \to \IR$ is a \emph{stable tail
    dependence function (stdf)} if and only if
  \begin{enumerate}[label=(\roman*), labelwidth=\widthof{(iii)}]
  \item $\STDF$ is \emph{homogeneous of degree $1$}, that is
    $\STDF(c\bm{x})=c\STDF(\bm{x})$, $\bm{x}\in[0,\infty)^d$, for all $c>0$;
  \item $\STDF(\bm{e}_j)=1$ for each unit vector
    $\bm{e}_j=(0,\dots,0,1,0,\dots,0)\in\IR^d$, $j\in\{1,\dots,d\}$; and
  \item\label{def:stdf:iii} $\STDF$ is \emph{fully
      $d$-alternating\footnote{
While \cite{Ressel2013} uses the term \emph{fully
          $d$-max-decreasing}, the terminology has later been changed to \emph{fully
      $d$-alternating}, see \cite{Ressel2022}, which we adopt in this article.}}, that is for all $0\le k\le d-1$ and
    any $k$ variables fixed, $\Delta_{(\bm{y},\bm{x}]}\STDF\le 0$,
    $\bm{x}\le\bm{y}$, where
    $\Delta_{(\bm{y},\bm{x}]}\STDF=\sum_{\bm{i}\in\{0,1\}^d}(-1)^{\sum_{j=1}^di_j}\STDF(y_1^{i_1}x_1^{1-i_1},\dots,y_d^{i_d}x_d^{1-i_d})$
    denotes the $\STDF$-\emph{volume}.
  \end{enumerate}
\end{definition}
Note that since
\begin{align*}
  \Delta_{(\bm{y},\bm{x}]}\STDF&=\sum_{\bm{i}\in\{0,1\}^d}(-1)^{\sum_{j=1}^d(1-i_j)}\STDF(y_1^{1-i_1}x_1^{i_1},\dots,y_d^{1-i_d}x_d^{i_d})\notag\\
  &=(-1)^d\sum_{\bm{i}\in\{0,1\}^d}(-1)^{\sum_{j=1}^d i_j}\STDF(x_1^{i_1}y_1^{1-i_1},\dots,x_d^{i_d}y_d^{1-i_d})=(-1)^d\Delta_{(\bm{x},\bm{y}]}\STDF,
\end{align*}
we can replace~\ref{def:stdf:iii} by the slightly more intuitive (iii') for all
$0\le k\le d-1$ and any $k$ variables fixed,
$(-1)^{d}\Delta_{(\bm{x},\bm{y}]}\STDF\leq 0$, $\bm{x}\le\bm{y}$.  Furthermore,
see \cite[Theorem~2.3.3, p.~111; Corollary~3.1.6, p.~141]{Falk2019} for a
textbook presentation, and \cite{hofmann2009} for the initial analytical
characterization relevant to this case, $\ell:[0,\infty)^d\to[0,\infty)$ is a
stdf if and only if it is a \emph{D-norm}, that is
$\ell(\bm{x})=\lVert \bm{x}\rVert_{\text{D}}=\E[\max_{1\le j\le
  d}\{\abs{x_j}W_j\}]$, $\bm{x}\in[0,\infty)^d$, for a \emph{generator}
$\bm{W}=(W_1,\dots,W_d)$, which is a random vector with $W_j\ge 0$ and
$\E[W_j]=1$, $j=1,\dots,d$.  For further properties of D-norms, see
\cite{aulbachfalkzott2015}, while \cite{hoferthuserprasad2018} provide an
overview of generators for well known stdfs.

The link between EVCs and stdfs is now that a copula $\Copula$ is an EVC
if and only if
\begin{align}
  \Copula(u_1,\dots,u_d) = \e^{-\STDF(-\log u_1,\dots,-\log u_d)},\quad\bm{u}\in[0,1]^d,\label{eq:repr:EVC}
\end{align}
for some stdf $\STDF$; see \cite{deheuvels1978} or \cite[Theorem~6.2.2,
p.~129]{GudendorfSegers2010}. By Sklar's theorem, one therefore also obtains
that a $d$-variate distribution function $\EVDist$ is max-stable if and only if
$\EVDist(\bm{x}) = \e^{-\STDF(-\log H_1(x_1),\dots,-\log H_d(x_d))}$,
$\bm{x}\in\IR^d$, for $H_j\in\GEV_1$, $j\in\{1,\dots,d\}$, and some stdf
$\STDF$.

We finish this section with basic properties of stdfs, with proofs for the sake of completeness.
The lemma is then used in the following section.
\begin{lemma}[Properties of stdfs]\label{lem:prop:stdf}
  Let $\STDF$ be a stdf with generator $\bm{W}$.
  \begin{enumerate}
  \item\label{lem:prop:stdf:1}
    $\max\{x_1,\dots,x_d\}\le\STDF(\bm{x})\le\sum_{j=1}^dx_j$. The lower
    bound is precisely the stdf of the comonotone copula $M(\bm{u})=\min_{1\le j\le d}\{u_j\}$
    and the upper bound is precisely the stdfs of the independence
    copula $\Pi(\bm{u})=\prod_{j=1}^du_j$.
  \item\label{lem:prop:stdf:2} $\STDF(\bm{0})=0$, $1\le\STDF(\bm{1})\le d$ and any value in $[1,d]$
    is attainable by a suitable stdf $\STDF$.
  \end{enumerate}
\end{lemma}
\subsection{Morillas-type transformations of stable tail dependence functions}\label{sec_trafo}
While Section~\ref{sec_cdf_distortions} shows how pre- and post-compositions can
be used to transform copulas, \cite[Proposition~2.6.1,
p.~132]{Falk2019} and \cite[Lemma~2.3\,(ii), p.~1029]{ChatelainFougeresNeslehova2020}
showed that
\begin{align}
  \ell_\gamma(\bm{x}):=\bigl(\STDF\bigl(x_1^{1/\gamma},\ldots,x_d^{1/\gamma}\bigr)\bigr)^\gamma,\quad\bm{x}\ge\bm{0},\label{eq_ChatelainFougeresNeslehova2020}
\end{align}
yields a valid stdf for any stdf $\STDF$ as long as $\gamma\in(0,1]$. The
transformation \eqref{eq_ChatelainFougeresNeslehova2020} is reminiscent of the
Morillas-type copula-to-copula transform in \eqref{eq_MorillasDef}, albeit for
stdfs. This naturally raises the question whether there are other Morillas-type
transforms than power functions $x^{1/\gamma}$, $\gamma\in(0,1]$, which can
transform stdfs into stdfs.
In the context of model identifiability for Archimax copulas, \cite[Lemma~2.4, p.~1029]{ChatelainFougeresNeslehova2020} have similarly singled out the pivotal role of power distortions of stdf's with (rational) exponents in $(0,1]$.
In this situation, different Archimedean generators lead to the same Archimax copula if and only if the associated stdf's are linked via a power distortion.
Our second main contribution in this work,
Theorem~\ref{thm:power:distortions:of:stdfs}, likewise shows that in the context of Morillas-type transformations of stdfs, \emph{only} power transformations guarantee that the resulting
function is again a stdf.
\begin{theorem}[Characterization of stdf to stdf transformations]\label{thm:power:distortions:of:stdfs}
  Denote by $g \colon [0,\infty) \to [0,\infty)$ a continuous and strictly
  increasing function such that $g(x)\to\infty $ for $x\to\infty$.
  For a stdf $\STDF$, let $\STDF_g$ be defined as
  \begin{align*}
    \STDF_g(\bm{x}) := g\bigl(\STDF(g^{-1}(x_1),\ldots,g^{-1}(x_d))\bigr), \quad \bm{x}=(x_1,\dots,x_d)\in[0,\infty)^d
  \end{align*}
where $g^{-1}$ is the inverse function to $g$.
  Then $\STDF_{g}$ is a stdf for all
  $\STDF$ if and only if $g(x) = cx^{\gamma}$, $x\ge 0$, for $\gamma\in(0,1]$
  and $c\in(0,\infty)$.
\end{theorem}

Although~\eqref{eq_ChatelainFougeresNeslehova2020}
resembles~\eqref{eq_MorillasDef}, it is important to note that stdfs are defined
with tighter restrictions and so the results of \cite{Morillas2005} in the
context of copulas do not directly carry over to our extreme value setting in
terms of stdfs. In the latter setting, this motivates the question which
transformations of Morillas-type guarantee that EVCs are
transformed into EVCs. By
Theorem~\ref{thm:power:distortions:of:stdfs}, the choice of stdf-preserving
transformations $g$ is narrowed down to power functions by the fact that
$\STDF_g$ is required to be a stdf for \emph{every} stdf $\ell$. For specific stdfs
$\STDF$, it is possible to find admissible transformations $g$ such that
$\STDF_g$ is again a valid stdf even if $g$ is not a power function or even not
a Morillas-type transformation, see the following example.
\begin{example}[Non-power and non-Morillas-type stdf-preserving transformations]
  \begin{enumerate}
  \item\label{ex:nonpower:trafo} For the comonotone stdf $\ell_M(\bm{x})=\max\{x_1,\dots,x_d\}$, any increasing and invertible function $g$ satisfies
    \begin{align*}
      \STDF_g(\bm{x})&=g(\max\{g^{-1}(x_1),\ldots,g^{-1}(x_d)\})=g(g^{-1}(\max\{x_1,\dots,x_d\}))=\ell_M(\bm{x}),\quad\bm{x}\ge\bm{0},
    \end{align*}
    and so maps the comontone stdf to itself. As a concrete example of such a
    function $g$ that is not a power function, take $g(x) = \e^x$.
  \item Consider stdfs with generators with equal marginal distributions $F_W$,
    where $F_W$ is the distribution function of
    $\Gamma(1-1/\theta)W\sim \e^{-x^{\theta}}$, $x\in[0,\infty)$,
    $\theta\in[1,\infty)$, for $\Gamma$ denoting the gamma function. A
    stochastic map, transforming comonotone $(W,\dots,W)$ with margins $F_W$ to
    $(W_1,\dots,W_d)$ with $W_1,\dots,W_d\isim F_W$, maps the comonotone stdf
    $\ell_M$ to the \emph{logistic stdf}
    $\STDF(\bm{x})=(\sum_{j=1}^d x_j^{\theta})^{1/\theta}$,
    $\bm{x}\in[0,\infty)^d$, with parameter $\theta\in[1,\infty)$.
    This is a valid stdf but,
    by~\ref{ex:nonpower:trafo}, cannot be of Morillas-type. %
  \end{enumerate}
\end{example}

Theorem~\ref{thm:power:distortions:of:stdfs} can be used to construct stdfs from
given elementary ones. For example, pairing $g(x)=x^{\gamma}$ with the stdf
$\STDF(\bm{x}) = \sum_{i=1}^d x_i$ of the independence copula, we obtain the
logistic stdf $\STDF_g(\bm{x})=(\sum_{i=1}^d x_i^{1/\gamma})^\gamma$
with parameter $\gamma\in(0,1]$. A more general example is the following.
\begin{example}[Constructing EVCs via distortions of stdfs]
  For $k=1,\dots,K$, let $\alpha_k>0$ and $C_k$ be an EVC with stdf
  $\STDF_k$. Clearly, as a mixture, $C=\sum_{k=1}^K\alpha_k C_k$ is a
  copula if $\sum_{k=1}^{K}\alpha_k=1$, but $C$ generally fails to satisfy max-stability and is thus not
  an EVC anymore. However, mixing at the level of the stdfs will result in an
  EVC.  To this end, note that the convex combination
  $\STDF=\sum^{K}_{k=1}\alpha_k \STDF_k$ of stdfs is, even without the condition $\sum_{k=1}^{K}\alpha_k=1$, readily seen to fulfill
  the defining properties of a stdf %
  and thus the associated copula $C_{\STDF}$ is again an EVC. It is given by
  $\Copula_{\STDF}(\bm{u}) = \prod_{k=1}^K
  \Copula_{\STDF_k}(u_1^{\alpha_k},\dots,u_d^{\alpha_k})$, which is a
  \emph{Liebscher transformation} of
  $\Copula_{\STDF_1},\dots,\Copula_{\STDF_K}$; see \cite{liebscher2008}.
  Theorem~\ref{thm:power:distortions:of:stdfs} now allows us to break the simple
  product structure of $\Copula_{\STDF}$. Power-transforming $\STDF$, we obtain
  the (by Theorem~\ref{thm:power:distortions:of:stdfs} valid) stdf
  \begin{align*}
    \STDF_g(\bm{x}) = g\biggl(\,\sum^{K}_{k=1}\alpha_k \STDF_k(g^{-1}(x_1),\dots,g^{-1}(x_d))\biggr) = \biggl(\,\sum^{K}_{k=1}\alpha_k \STDF_k(x_1^{1/\gamma},\dots,x_d^{1/\gamma})\biggr)^{\gamma},\quad\bm{x}\in[0,\infty)^d,
  \end{align*}
  where $\gamma\in(0,1]$. The EVC based on $\STDF$ then takes the form
  \begin{align*}
    \Copula_{\STDF}(\bm{u}) = \e^{-\left(\sum^{K}_{k=1}\alpha_k \STDF_k\left((-\log u_1)^{1/\gamma},\dots,(-\log u_d)^{1/\gamma}\right)\right)^{\gamma}},\quad\bm{u}\in[0,1]^d.
  \end{align*}
  Additionally distorting $\STDF_k$ by $\gamma_k\in(0,1]$, $k=1,\dots,K$, by
  Theorem~\ref{thm:power:distortions:of:stdfs}, we obtain the even more flexible stdf
  \begin{align*}
    \STDF(\bm{x}) = \left(\,\sum^{K}_{k=1}\alpha_k \STDF_k\left(x_1^{1/(\gamma_k\gamma)},\dots,x_d^{1/(\gamma_k\gamma)}\right)^{\gamma_k}\right)^{\gamma},\quad\bm{x}\in[0,\infty)^d,
  \end{align*}
  with corresponding EVC
  \begin{align*}
    \Copula_{\STDF}(\bm{u}) = \e^{-
    \left(\sum^{K}_{k=1}\alpha_k \STDF_k\left((-\log u_1)^{1/(\gamma_k\gamma)},\dots,(-\log u_d)^{1/(\gamma_k\gamma)}\right)^{\gamma_k}\right)^{\gamma}},\quad\bm{u}\in[0,1]^d.
  \end{align*}
  Such constructions may serve as flexible models for multivariate extremes and
  have appeared, for specific choices of $\ell_1,\dots,\ell_K$, in \cite{mcfadden1978}, \cite{tawn1990},
  \cite{stephenson2003} and \cite{stephenson2009} to this end.
\end{example}

Constructing stdfs from given elementary ones is reminiscent of the construction of
Archimedean generators from elementary generators via power and other transformations; see
\cite[Theorem~4.5.1]{Nelsen2006} and \cite[Chapters~2 and 4]{Hofert2010}.
Furthermore, the
following result, which is reminiscent of a property of Archimedean generators, shows that the stdf $\STDF_g$ in
Theorem~\ref{thm:power:distortions:of:stdfs} is only uniquely determined up to a
scaling factor $c>0$ (a nuisance parameter in terms of the dependence structure
induced by $\STDF_g$).
\begin{corollary}[Uniqueness up to a scaling factor]\label{cor_uniqueness_up_to_scaling}
  With $g$ as in Theorem~\ref{thm:power:distortions:of:stdfs}, $g(cx)$
  generates the same stdf $\STDF_g$ for all $c>0$.
\end{corollary}
In order to highlight the fact that the specific power function form in
Theorem~\ref{thm:power:distortions:of:stdfs} arises from the constraint that the
resulting $\STDF_g$ has to be a valid stdf, we now consider the case of Archimax
copulas where this constraint can be weakened.
\begin{remark}[Connection with Archimax copulas]
  The restriction that $g$ is a power function in
  Theorem~\ref{thm:power:distortions:of:stdfs} is owed to the fact that $g$ has
  to map stdf's to stdf's.  When relaxing this constraint on the domain,
  other functional forms leading to valid dependence structures are possible.
  For example, as already noted in \cite[p.~395]{GenestRivest2001},
  Morillas-type distortions of EVCs by distortions linked to Archimedean
  generators lead to Archimax copulas.  Specifically, a function
  $\gen:[0,\infty)\to[0,1]$ is \emph{$d$-monotone} if it is continuous on
  $[0,\infty)$, admits derivatives up to the order $d-2$ satisfying
  $(-1)^k\gen^{(k)}(t)\ge 0$ for all $k\in\{0,\dots,d-2\}$, $t\in(0,\infty)$,
  and $(-1)^{d-2}\gen^{(d-2)}(t)$ is decreasing and convex on $(0,\infty)$; see
  \cite{McNeilNeslehova2009}. By the same reference, $\gen$ is a \emph{$d$-dimensional
    Archimedean generator} if it is $d$-monotone, $\lim_{t\to\infty}\gen(t)=0$
  and $\gen(0)=1$. As noted in \cite[Theorem~2.2,
  p.~3065]{McNeilNeslehova2009}, $d$-dimensional Archimedean generators stand in
  direct correspondence with $d$-dimensional Archimedean copulas.  For a given
  Archimedean generator, we can thus define $g(x) := -\log \gen(x)$, $x\geq 0$,
  with inverse $g^{-1}(x) = \invgen\left(\e^{-x}\right)$.  For an EVC $\Copula$
  with associated stdf $\STDF$, this leads to the Morillas-type transformation
  \begin{align}\label{eq_stdf_g}
    \STDF_g(\bfx) = -\log\left(\gen(\STDF\left(\invgen\left(\e^{-x_1}\right),\dots,\invgen\left(\e^{-x_d}\right)\right)\right), \quad \bfx = (x_1,\dots,x_d)\in[0,\infty)^d.
  \end{align}
  While for the Gumbel generator $\gen(t) = \exp(-t^{1/\theta})$,
  $\theta \geq 1$, $g$ will be an admissible power function in the sense of
  Theorem~\ref{thm:power:distortions:of:stdfs}, this will in general not be the
  case anymore, so $\STDF_g$ in \eqref{eq_stdf_g} will in general not be a
  stdf.  However, utilizing $\STDF_g$ nonetheless in an EVC type construction as
  in \eqref{eq_max_stable} allows us to define a function $\Copula_g$ via
  $\Copula_g(\bm{u}):= \e^{-\STDF_g(-\log u_1,\dots,-\log u_d)}=\gen\left(\STDF\left(\invgen\left(u_1\right),\dots,\invgen\left(u_d\right)\right)\right)$, $\bm{u}\in[0,1]^d$.
  Here $\Copula_g$ is readily identified as an Archimax copula, see
  \cite{CharpentierFougeresGenestNeslehova2014}. Hence, while only the Gumbel
  generator linked to power functions will guarantee that $\STDF_g$ is a valid
  stdf for any stdf $\STDF$ and therefore that $\Copula_g$ is an EVC, more
  general transformations can still lead to valid dependence structures outside
  the EVC class.
\end{remark}

We can now turn to the aforementioned point and apply
Theorem~\ref{thm:power:distortions:of:stdfs} to identify the distortions $\pD$
such that the Morillas-type transformation in~\eqref{eq_MorillasDef} transforms
EVCs into EVCs; see
\cite{HerrmannHofertNeslehova2024} for an alternative proof.
\begin{corollary}[EVC to EVC transformations]\label{cor_EVC_to_EVC}
  Let $\pD$ be a continuous and strictly increasing distortion function. Then
  $\Copula_{\pD}$ defined in \eqref{eq_MorillasDef} is a $d$-dimensional EVC
  for any $d$-dimensional EVC $\Copula$ if and only if
  $\pD(u) = \e^{-\lambda(-\log u)^\gamma}$, $u\in[0,1]$, for $\lambda>0$ and $\gamma\in(0,1]$.
\end{corollary}
Exactly as the scaling $c>0$ in Corollary~\ref{cor_uniqueness_up_to_scaling},
the parameter $\lambda>0$ of the distortions considered in
Corollary~\ref{cor_EVC_to_EVC} is a nuisance parameter in the sense that
different values of $\lambda$ generate the same EVC.

\section{Maximum domain of attraction for Morillas-type copula
  distortions}\label{sec_MDA_for_C_D}
\subsection{Maximum domain of attraction for $\Copula_{\pD}$}
Section~\ref{sec_basic_notions} is centered around understanding transformations
of the possible limit laws (copulas or not) of componentwise multivariate
extremes. In this section we leverage these insights to obtain results from the
more realistic starting point of the joint distribution generating the original
observations. The corresponding initial dependence structure and the dependence
structure of the limiting distribution are linked by the copula MDA as laid out
in Definition~\ref{def_extreme_value_copula}.  We therefore investigate the MDA
for the distorted copula $\Copula_{\pD}$ and show in our third main result
Theorem~\ref{theorem_max_domain_of_attraction_Morillas} below how the limiting
extreme value copulas of $\Copula$ and $\Copula_{\pD}$ are connected.

As preparation of the proof of
Theorem~\ref{theorem_max_domain_of_attraction_Morillas} we have the following
Lemma~\ref{rate_function_convergence}, which highlights that the rate used in
the MDA condition in Definition~\ref{def_extreme_value_copula} is a sufficient
criterion concerning the convergence with respect to other rates. Note that
Item~\ref{rate_function_convergence_IV} is directly taken from
\cite[Théorème~(3.1), p.~14]{deheuvels1978} and added for the sake of completeness.
\begin{lemma}[Alternative rates for the MDA of copulas]\label{rate_function_convergence}
  Let $\Copula$ be a $d$-copula and let $\Copula^*\in\SetEVTCopula_d$ be an EVC.
  Then the following statements are equivalent:
\begin{enumerate}
\item\label{rate_function_convergence_I} $\Copula\in\MDA(\Copula^*)$.
\item\label{rate_function_convergence_II} For all rate functions
  $\rate \colon \IN \to (0,\infty)$, $n\mapsto\rate_n$ such that
  $\rate_n \to \infty$ as $n \to \infty$, we have
  $\Copula^{\rate_n}\bigl( u_1^{1/\rate_n}, \ldots, u_d^{1/\rate_n}\bigr) \to \Copula^*(u_1, \ldots, u_d)$,
  $(u_1,\dots,u_d)\in[0,1]^d$.
\item\label{rate_function_convergence_III} There exists a rate function
  $\rate\colon \IN \to (0,\infty)$, $n\mapsto\rate_n$ such that $\rate_n$ is
  asymptotically equivalent to $n$, that is
  $\lim_{n\to\infty}n/\rate_n = \theta$ for some $\theta>0$, and
  $\Copula^{\rate_n}\bigl(u_1^{1/\rate_n}, \ldots, u_d^{1/\rate_n}\bigr) \to \Copula^*(u_1, \ldots, u_d)$,
  $(u_1,\dots,u_d)\in[0,1]^d$.
\item\label{rate_function_convergence_IV} There exists an increasing sequence of
  indices $(n_k)_{k\geq 1}$ such that $\{n_{k+1}/n_k \colon k\geq 1\}$ is
  bounded and
  $\Copula^{n_k}\bigl(u_1^{1/n_k},\dots,u_d^{1/n_k}\bigr)\to\Copula^*(u_1,\dots,u_d)$
  for all $(u_1,\dots,u_d)\in[0,1]^d$.
\end{enumerate}
\end{lemma}
It is important to note that the `for all' in \ref{rate_function_convergence_II}
of Lemma~\ref{rate_function_convergence} cannot be replaced with `there exits at
least one'.  For example, \cite[Example~4.1, p.~430]{KlementMesiarPap2005}
provide an example of a copula $\Copula_0$ and two rate functions
$\rate_n = 2^n$ and $\tilde{\rate}_n = 3\cdot 2^n$ such that
\begin{align*}
\lim_{n\to\infty}
  \Copula_0^{\rate_n}(u^{1/\rate_n},v^{1/\rate_n}) = \Copula_0(u,v)\quad\text{and}\quad
  \lim_{n\to\infty}\Copula_0^{\tilde{\rate}_n}(u^{1/\tilde{\rate}_n},v^{1/\tilde{\rate}_n}) = \Copula_0^3(u^{1/3},v^{1/3}).
\end{align*}
Thus $\Copula_0$ is not in the MDA of \emph{any} EVC since the sequence
$(\Copula_0^{n}(u^{1/n},v^{1/n}))_n$ has two distinct
limit points and hence does not converge; \cite[Exemple~3.1,
p.~14]{deheuvels1978} contains another such example. Rephrased in the context of
Lemma~\ref{rate_function_convergence}, this implies that it is not possible to
conclude from the convergence
$\lim_{n\to\infty}\Copula_0^{\rate_n}(u^{1/\rate_n},v^{1/\rate_n}) =
\Copula^*(u,v)$ for \emph{one specific} rate function $\rate_n$ (other than
those which are asymptotically equivalent to $n$ as established in
Lemma~\ref{rate_function_convergence}~\ref{rate_function_convergence_III}) that
$\Copula_0\in\MDA(\Copula^*)$. An intuitive reason for rate functions
asymptotically equivalent to $n$ directly allowing this conclusion is that in the
limit, when $\rate_n \approx n/\theta$, the limiting extreme value copula is
invariant under the additional power transformation $1/\theta$ and hence all
limits are in agreement.

Equipped with Lemma~\ref{rate_function_convergence} we can now proceed towards the third
main result of this work.
The following Theorem~\ref{theorem_max_domain_of_attraction_Morillas} provides
sufficient conditions for the distorted copula $\Copula_{\pD}$ to fall into the
MDA of an EVC, and clearly identifies this limiting EVC.

In the following we denote by $\EW_{\gamma}$ the \emph{extremal Weibull} distribution with shape parameter $\gamma>0$, that is
\begin{align*}
\EW_{\gamma}(x) = \begin{cases}
  \exp\left(-(-x)^{\gamma}\right), & x\leq 0,\\
  1, & x > 0.
\end{cases}
\end{align*}
As usual, a univariate distribution function $\pF$ is in the \emph{maximum domain of attraction (MDA)}
of a non-degenerate distribution function $H$ (short: $\pF\in\MDA(H)$) if there exists $\an_n>0$, $\bn_n\in\IR$, $n\in\IN$,
such that $\lim_{n\to\infty} \pF^n(\an_n x+\bn_n)=H(x)$, $x\in\IR$; see \cite[Definition~3.3.1, p.~128]{embrechtsklueppelbergmikosch1997}. %
Lastly, the upper endpoint of the support of $\pF$ is denoted by $\support_+(\pF) := \sup\{x\in\IR\colon\pF(x)<1\}$.
\begin{theorem}[MDA for Morillas-type copula distortions]\label{theorem_max_domain_of_attraction_Morillas}
  Let $\Copula$ be a $d$-copula in the MDA of an EVC
  $\Copula^*$, and by $\pD$ a distortion function such that $\pD\in\MDA(\EW_{\gamma})$, $\gamma>0$, and $\support_+(\pD) = 1$.
  \begin{enumerate}
  \item\label{item_distortion_theorem_1} If $\pD$ is such that $\Copula_{\pD}$
    is a $d$-copula, then $\Copula_{\pD}$ is in the MDA of
    $\Copula^*_{\pD^*}$, where $\pD^*(u)=\e^{-(-\log u)^{\gamma}}$, $u\in[0,1]$, is a distortion function.
  \item\label{item_distortion_theorem_2} If $\pD\in\AMDCO{d}$ then
    \ref{item_distortion_theorem_1} holds with the restriction $\gamma\in(0,1]$.
    \end{enumerate}
    In both cases $\gamma=1$ implies $\Copula^*_{\pD^*} = \Copula^*$.
\end{theorem}
In the proof of Theorem~\ref{theorem_max_domain_of_attraction_Morillas} the
convergence $\pD^{n}(u^{1/{\rate_n}}) \to \pD^*(u)$ for some rate function
$\rate_n$ and limiting distortion $\pD^*$ is crucial. For random variables
$V_1,\dots,V_n\isim\pD$ we can alternatively write
\begin{align*}
\pD^{n}(u^{1/{\rate_n}}) = \Prob(\max(V_1,\dots,V_n) \leq u^{1/\rate_n}),
\end{align*}
so that the assumed convergence falls under the framework of \cite{Pancheva1984} for
power stabilizations of extremes; see also \cite{MohanRavi1993} and \cite[Section~2.6, p.~37~ff.]{BarakatKhaledNigm2019} for
a detailed discussion of this case from the power stabilization point-of-view.
As discussed in the proof, in our specific case power stabilizations are however equivalent to more commonly considered affine stabilizations, see \cite[Theorem~3.1~(e), p.~636]{MohanRavi1993}, allowing us to state our assumptions in terms of the more familiar Weibull MDA.

The diagram in Figure~\ref{fig_MDA_diagram} depicts the structure implied by
Theorem~\ref{theorem_max_domain_of_attraction_Morillas}.
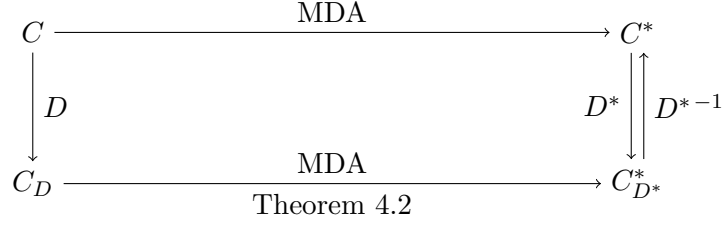
\begin{figure}
  \label{fig_MDA_diagram}
  \centering
  \begin{tikzpicture}[auto]
    \node (P) at (0,0) {$\Copula$};
    \node(Q) at (8,0) {$\Copula^*$};
    \node (B) at (0,-2) {$\Copula_{\pD}$};
    \node (C) at (8,-2) {$\Copula^*_{\pD^*}$};
    \draw[transform canvas={xshift=0.5ex},->] (C) - -(Q) node[right,midway] {$\qDs$};
    \draw[transform canvas={xshift=-0.5ex},->](Q) -- (C) node[left,midway] {$\pD^*$};
    \draw[->](P) to node {$\MDA$}(Q);
    \draw[->] (P) to node {$\pD$} (B);
    \draw[->] (B) - -(C) node[above,midway] {$\MDA$} node[below,midway] {Theorem~\ref{theorem_max_domain_of_attraction_Morillas}};
  \end{tikzpicture}
  \caption{Diagram depicting the structure implied by Theorem~\ref{theorem_max_domain_of_attraction_Morillas}.}
\end{figure}
While the proof of Theorem~\ref{theorem_max_domain_of_attraction_Morillas}~\ref{item_distortion_theorem_2} crucially utilizes the results derived from the investigation of properties of
stdfs in Theorem~\ref{thm:power:distortions:of:stdfs} and
Corollary~\ref{cor_EVC_to_EVC}, stdfs based techniques are not appropriate to
solve the general MDA question since they are by
definition limited only to the context of EVCs.

As an example where \ref{item_distortion_theorem_1} in Theorem~\ref{theorem_max_domain_of_attraction_Morillas} holds for $\pD\not\in\AMDCO{d}$, consider the comonotonic copula $M$ and any bijective distortion $\pD\not\in\AMDCO{d}$ such that $\pD^n(u^{1/\rate_n})\to \pD^{*}(u) = \e^{-(-\log u)^{\gamma}}$, where $\gamma >0$ for some rate $\rate_n$.
In this case $M_D = M_{\pD^{*}} = M$ and hence \ref{item_distortion_theorem_1} trivially  holds, showing that the $d$-absolute monotonicity condition in \ref{item_distortion_theorem_2} is merely sufficient.
The example also provides an instance where we can have $\gamma\neq 1$ and yet $\Copula^*_{\pD^*} = \Copula^*$.
As a concrete example consider the Kumaraswamy distribution with
distribution function $K_{a,\gamma}(x) = 1 - (1 - x^a)^{\gamma}$, $a,\gamma>0$, where we have
$K_{a,\gamma}\in\MDA(\EW_{\gamma})$; see Example~\ref{Kumaraswamy_distortion}.  As such,
any $\gamma\in(0,\infty)$ can be reached for the limit
$\pD^*(u)=\e^{-(-\log u)^{\gamma}}$ in Theorem~\ref{theorem_max_domain_of_attraction_Morillas} and
there is no a-priori restriction of $\gamma$ to $(0,1]$ since the
Kumaraswamy distribution is not $d$-absolutely monotone in general.
It is the additional hypothesis of $d$-absolute monotonicity in Theorem~\ref{theorem_max_domain_of_attraction_Morillas}~\ref{item_distortion_theorem_2} that is in fact responsible
for restricting the possible values of $\gamma\in(0,\infty)$ to $(0,1]$.

Theorem~\ref{theorem_max_domain_of_attraction_Morillas} also has a direct important implication in case the initial copula $\Copula$ is in
the MDA of the Gumbel--Hougaard copula, that is
$\Copula\in\MDA(\GHCopula_{\theta})$, which specifically includes the
independence copula $\Pi$.  Here
$\GHCopula_{\theta}(\bfu) = \exp\bigl(-(\sum_{j=1}^{d} (-\log
u_j)^{\theta})^{1/\theta}\bigr)$ denotes the $d$-dimensional Gumbel--Hougaard
copula with parameter $\theta\geq 1$.  For the Gumbel--Hougaard copula, we have
for $\pD^*(u)=\e^{-(-\log u)^{\gamma}}$ with $\gamma\in(0,\theta]$ that
$(\GHCopula_{\theta})_{\pD^*} = \GHCopula_{\theta/\gamma}$, which shows that
this class of copulas is closed under distortions of this type.  We summarize
this logic in the following corollary.
\begin{corollary}[Distortions in the Gumbel--Hougaard MDA
  remain in the Gumbel--Hougaard MDA]\label{cor_Pi_to_Gumbel}
  For a $d$-copula $\Copula$ such that $\Copula\in\MDA(\GHCopula_{\theta})$ and $\pD\in\MDA(\EW_{\gamma})$, $\gamma\in(0,\theta]$, with $\support_+(\pD)=1$ such that $\Copula_{\pD}$ is a copula (for example if $\pD\in\AMDCO{d}$), we have $\Copula_\pD\in\MDA(\GHCopula_{\theta/\gamma})$. In particular, for
  $\Copula\in\MDA(\Pi)$ we have $\Copula_\pD\in\MDA(\GHCopula_{1/\gamma})$.
\end{corollary}
As a consequence, since $\gamma\in(0,1]$ for $\pD\in\AMDCO{d}$ according to
Theorem~\ref{theorem_max_domain_of_attraction_Morillas}~\ref{item_distortion_theorem_2},
a suitable $d$-absolutely monotone distortion never decreases upper-tail
dependence for any $\Copula \in \MDA(\GHCopula_{\theta})$; see
\cite[Section~5.4, pp.~214]{Nelsen2006} for the relevant definitions and a
derivation of the fact that the upper tail dependence coefficient of the Gumbel
copula $\GHCopula_{\theta}$ is given as $2-2^{1/\theta}$.  In case of
$\Copula \in \MDA(\Pi)$, asymptotic independence is even transformed into
Gumbel--Hougaard type extremal dependence.  This effect is visible in
Figure~\ref{fig_Clayton_zeta}, where the initial Clayton copula is in the MDA of
$\Pi$ while the transformed copula clearly exhibits upper tail dependence.
According to Corollary~\ref{cor_Pi_to_Gumbel} the limiting Gumbel--Hougaard
copula has in this case a parameter value of $2$ since the utilized zeta
distortion satisfies $\pD_s\in\MDA(\EW_{0.5})$ for the chosen parameter $s=1.5$,
see Example~\ref{example_zeta_distortion}, and hence the upper tail dependence
coefficient changes from zero in the undistorted model to
$2-\sqrt{2} \approx 0.5857864$ in the distorted model.  As a second example, the
normal copula is in the MDA of $\Pi$, and therefore any
$d$-absolutely monotone distortion leads to the MDA of
the Gumbel--Hougaard copula.  More importantly, since any Archimedean copula
$\Copula_{\gen}$ is formally a Morillas-type transformation of the independence
copula $\Pi$, see \cite[p.~170]{morillas2005b}, \cite[Corollary~7,
p.~5]{durrleman2000simple} and \cite[p.~395]{GenestRivest2001}, we know that
$\Copula_{\gen}=(\Pi)_{\pD_\gen}$ is in the MDA of the
Gumbel--Hougaard \emph{as long as} the associated distortion $\pD_{\gen}$
satisfies the assumptions of
Theorem~\ref{theorem_max_domain_of_attraction_Morillas}.  While it is therefore
at this moment not clear if our discussion covers all Archimedean copulas (we
provide details pertaining to the Archimedean case in
Section~\ref{sec_Archimedean_connection}), this finding re-affirms the fact that
Archimedean copulas are typically in the MDA of the
Gumbel--Hougaard copula; see \cite[Statement~B, p.~208]{GenestRivest1989}.  We
provide additional arguments to fully recover Statement B of
\cite{GenestRivest1989} in Corollary~\ref{cor_Archimedean_max_domain} below.
\subsection{Choice of the rate function $\rate_n$: Asymptotic equivalence}
The condition $\pD\in\MDA(\EW_{\gamma})$ with $\support(\pD)=1$ in Theorem~\ref{theorem_max_domain_of_attraction_Morillas} guarantees the convergence $\pD^n(u^{1/\rate_n})\to\pD^*(u)$ to a non-degenerate distortion $\pD^*$.
As established in \cite[Theorem~2.2, p.~634, Theorem~3.1~(e), p.~636]{MohanRavi1993} one can for this purpose utilize the rate function
\begin{align}\label{eq_rate}
\rate_n := \frac{1}{-\log\qD(1-1/n)},
\end{align}
making the results in Theorem~\ref{theorem_max_domain_of_attraction_Morillas} widely applicable.
While these type of stabilizations are crucial to Theorem~\ref{theorem_max_domain_of_attraction_Morillas}, it is however not
immediately clear if there can be different rate functions that stabilize the
same distortion function $\pD$ in the assumed fashion. While this can be the case, all stabilizing
rate functions leading to continuous limits are necessarily asymptotically equivalent, meaning that for any
two admissible rate functions $\rate_n$ and $\tilde{\rate}_n$, there is a
$\theta\in(0,\infty)$ such that
$\rate_n / \tilde{\rate}_n \to \theta\in(0,\infty)$ and the connection between
the respective limits $\pD^*$ and $\widetilde{\pD}^*$ is given
in~\eqref{eq_connection_limiting_distortions} below.

To see this consider rates $\rate_n$ and $\tilde{\rate}_n$ such that
$\pD^n(u^{1/\rate_n})\to\pD^*(u)$, and also
$\pD^n(u^{1/\tilde{\rate}_n})\to\widetilde{\pD}^*(u)$ where we assume that $\pD^*$ and $\widetilde{\pD}^*$ are continuous distortion functions with $\pD^*(0)=0=\widetilde{\pD}^*(0)$ and $\pD^*(1)=1=\widetilde{\pD}^*(1)$.
Given that $\pD^*$ and $\widetilde{\pD}^*$ are continuous, the assumed convergence is in fact uniform in both cases, see \cite[p.~1]{Resnick1987}.
If $\rate_n / \tilde{\rate}_n$ is
convergent in $[0,\infty]$, we first show that
$\rate_n / \tilde{\rate}_n \to \{0,\infty\}$ is not possible.  To this end,
assume first that $\rate_n / \tilde{\rate}_n \to 0$, which then implies for
$u\in[0,1]$ that $u_n := u^{\rate_n/\tilde{\rate}_n} \to \ind{(0,1]}(u)$.
With $\pD^*(0)=0$ and $\pD^*(1)=1$, the uniform convergence
$\pD^n(u^{1/\rate_n})\to\pD^*(u)$ implies that
\begin{align*}
  \widetilde{\pD}^*(u) =
  \lim_{n\to\infty} \pD^n(u^{1/\tilde{\rate}_n})
  = \lim_{n\to\infty} \pD^n\left(\left(u^{\rate_n/\tilde{\rate}_n}\right)^{1/\rate_n}\right) = \pD^*(\ind{(0,1]}(u)) = \ind{(0,1]}(u).
\end{align*}
However, this contradicts that $\widetilde{\pD}^*$ is continuous on $[0,1]$.  In
case $\rate_n / \tilde{\rate}_n \to \infty$ a similar calculation shows
$\widetilde{\pD}^*(u) = \ind{\{1\}}(u)$, $u\in[0,1]$, which again contradicts
the assumed continuity of $\widetilde{\pD}^*$. Now consider the case where
$\lim_{n\to\infty} \rate_n / \tilde{\rate}_n = \theta \in(0,\infty)$. Then
$u_n := u^{\rate_n/\tilde{\rate}_n} \to u^{\theta}$ and therefore
\begin{align}\label{eq_connection_limiting_distortions}
\widetilde{\pD}^*(u)
= \lim_{n\to\infty} \pD^n(u^{1/\tilde{\rate}_n})
= \lim_{n\to\infty} \pD^n\left(\left(u^{\rate_n/\tilde{\rate}_n}\right)^{1/\rate_n}\right)
 = \pD^*(u^\theta),
\end{align}
where, as before, the uniform convergence $\pD^n(u^{1/\rate_n})\to\pD^*(u)$ is
utilized. Note how the resulting relationship
$\widetilde{\pD}^*(u) = \pD^*(u^\theta)$ is perfectly in line with the theory
established in Theorem~\ref{theorem_max_domain_of_attraction_Morillas}~\ref{item_distortion_theorem_1} in the following sense:
Following Theorem~\ref{theorem_max_domain_of_attraction_Morillas}~\ref{item_distortion_theorem_1} we have on the one hand $\Copula_{\pD}\in\MDA(\Copula^*_{\widetilde{\pD}^*})$ when using $\tilde{\rate}_n$, while on the other hand simultaneously that $\Copula_{\pD}\in\MDA(\Copula^*_{\pD^*})$ when using $\rate_n$.
This apparent contradiction is resolved when noting that $\widetilde{\pD}^* = \pD^*(u^\theta) = \pD^* \circ (\cdot)^{\theta}$ and hence, using the distortion composition property mentioned in Section~\ref{sec_properties_distortions}, we have
$
\Copula^*_{\widetilde{\pD}^*}(\bfu) = (\Copula^*_{(\cdot)^{\theta}})_{\pD^*}(\bfu) = \Copula^*_{\pD^*}(\bfu)
$,
where $\Copula^*_{(\cdot)^{\theta}} = \Copula^*$ holds by virtue of $\Copula^*$ being an EVC.

In case $\rate_n / \tilde{\rate}_n$ is \emph{not} convergent in $[0,\infty]$,
the sequence $u_n := 2^{-\rate_n/\tilde{\rate}_n}$ is non-convergent and
bounded, %
which means there are at least two distinct limit points $c_1$ and $c_2$. We
distinguish two cases. First, at least one limit point is in the open
interval $(0,1)$, where we assume wlog $c_1\in(0,1)$, and second, that no limit
point is in $(0,1)$, where we then assume wlog that $c_1=1$ and $c_2=0$.

In the first case, there are sub-sequences $(u_{n_i})$ and $(u_{m_i})$ such that
$u_{n_i}\to c_1\in(0,1)$ and $u_{m_i}\to c_2\in[0,1]\setminus\{c_1\}$.
Furthermore, since $\pD^*$ is continuous, $\pD^*(0) = 0$ and $\pD^*(1)=1$, there
is an $u^*\in(0,1)$ and an $\varepsilon>0$ such that $\pD^*$ is strictly
increasing on $(u^*-\varepsilon,u^*+\varepsilon)$.  For
$\alpha = \log(u^*)/\log(c_1)$ we then have for $v := 2^{-\alpha}\in(0,1)$ and
$v_n = v^{\rate_n/\tilde{\rate}_n} = u_n^{\alpha}$ that
$v_{n_i} \to c_1^{\alpha} = u^*$.  Similarly $v_{m_i}\to c_2^{\alpha}\neq u^*$.
The uniform convergence $\pD^n(u^{1/\rate_n})\to\pD^*(u)$ now implies that
$\lim_{n_i\to\infty} \pD^{n_i}(v^{1/\tilde{\rate}_{n_i}}) = \lim_{n_i\to\infty}
\pD^{n_i}(v_{n_i}^{1/\rate_{n_i}}) = \pD^*(u^*)$ and that
$\lim_{m_i\to\infty} \pD^{m_i}(v^{1/\tilde{\rate}_{m_i}}) = \lim_{m_i\to\infty}
\pD^{m_i}(v_{m_i}^{1/\rate_{m_i}}) = \pD^*(c_2^{\alpha})$.  Given that $\pD^*$
is non-decreasing and strictly increasing at $u^*$ we have
$\pD^*(c_2^{\alpha}) \neq \pD^*(u^*)$. However, this implies that for infinitely
many $n$, $\pD^n(v^{1/\tilde{\rate}_{n}})$ is arbitrarily close to
$\pD^*(c_2^{\alpha})$ and also to $\pD^*(u^*)$ which contradicts the assumed
convergence
$\lim_{n\to\infty} \pD^n(v^{1/\tilde{\rate}_{n}}) \to \widetilde{\pD}^*(v)$.

In the second case, a similar argument shows that
$\pD^n(2^{-1/\tilde{\rate}_{n}})$ is arbitrarily close to $0$ and also
arbitrarily close to $1$, which again leads to a contradiction to the
convergence
$\lim_{n\to\infty} \pD^n(2^{-1/\tilde{\rate}_{n}}) \to \widetilde{\pD}^*(1/2)$.

In summary, this shows that rate functions, stabilizing in the discussed fashion
to a continuous limit, are asymptotically equivalent, which specifically covers
the case relevant to Theorem~\ref{theorem_max_domain_of_attraction_Morillas}.
If an alternative rate function $\tilde{\rate}_n$ was established leading to an
alternative limit $\widetilde{\pD}^*$ we necessarily have
$\rate_n / \tilde{\rate}_n \to \theta\in(0,\infty)$.  Furthermore, by
\eqref{eq_connection_limiting_distortions}, the limit $\widetilde{\pD}^*$ is
given by $\widetilde{\pD}^*(u) = \pD^*(u^{\theta}) = \e^{-(-\log u^{\theta})^{\gamma}} = \e^{-\theta^{\gamma}(-\log u^{\theta})^{\gamma}} = \e^{-\lambda(-\log u^{\theta})^{\gamma}}$,
where $\lambda = \theta^{\gamma}$.  The appearance of the additional parameter
$\lambda>0$ is not problematic.  On the one hand, the additional parameter is in
line with results established in Section~\ref{sec_trafo}, specifically
Corollary~\ref{cor_EVC_to_EVC}, which is utilized in the proof of
Theorem~\ref{theorem_max_domain_of_attraction_Morillas}~\ref{item_distortion_theorem_2}.
On the other hand, for the limiting EVCs, $\lambda$ is a nuisance
parameter. Indeed, a direct calculation shows that
$\Copula^*_{\pD^*} = \Copula^*_{\widetilde{\pD}^*}$ due to the fact that EVCs
are invariant under power transformations and
$\widetilde{\pD}^*(u) = \left(\pD^*(u)\right)^{\lambda}$.
\subsection{Maximum domain of attraction for $d$-absolutely monotone and absolutely monotone distortions}
In this section we focus on the MDA and related results for the central classes of $d$-absolutely monotone and absolutely monotone distortions.
Although distortion functions in the Weibull MDA constitute a
large class of distortion functions, making Theorem~\ref{theorem_max_domain_of_attraction_Morillas} widely applicable, there also exist distortions such as
$\pD(u) = 1-\e^{-cu/(1-u)}$, $c>0$, which belong to the MDA of the Gumbel distribution $\Lambda(x)=\exp(-\e^{-x})$, see for
example \cite[Example~3.3.22,
p.~139]{embrechtsklueppelbergmikosch1997}. However, the latter example fails to
be $2$-absolutely monotone since the stated $\pD$ is non-convex on the interval
$u\in(\max(1-c/2,0),1)$. This is not a peculiarity of this example, but instead
a general feature of distributions in the Gumbel MDA. To see this, recall
the following abbreviated version of \cite[Corollary~2.5.3, p.~100]{deHaan1970},
which states that for a cumulative distribution function $\pF$ with finite right
endpoint $x_+ < \infty$, we have that if $\pF\in\MDA(\Lambda)$, then
\begin{align*}
  \lim_{x\nearrow x_+} \frac{\log(1-\pF(x))}{\log(x_+ - x)} = \infty.
\end{align*}
This result directly leads to the following corollary, ruling out that
$d$-absolutely monotone distortions fall into the Gumbel MDA.
\begin{corollary}[Distortions approaching the upper support endpoint below the identity cannot be in the Gumbel MDA]\label{cor_d_abs_monotone_not_Gumbel}
If $\pD$ is a distortion function such that $\pD(u) \leq u$ for $u\in(1-\varepsilon,1)$ and some $\varepsilon>0$, then
  $\pD\not\in\MDA(\Lambda)$. In particular, no $\pD\in\AMDCO{d}$ can be in the
  Gumbel MDA for any $d\geq 2$.
\end{corollary}
While Corollary~\ref{cor_d_abs_monotone_not_Gumbel} rules out the Gumbel MDA for $d$-absolutely monotone distortions, the Fréchet MDA is equally ruled out due to
the finite right endpoint of the support of distortions, leaving the Weibull MDA as the only possible MDA.
To counter the intuition that therefore all admissible $d$-absolutely
monotone distortions belong to the Weibull MDA we
explicitly provide a strictly increasing $2$-absolutely monotone distortion
outside the regular variation framework in the following example.
We then move on to discuss sufficient conditions for a $d$-absolutely monotone distortion to actually fall into the Weibull MDA.
\begin{example}[A $2$-absolutely monotone distortion without regularly varying tails]\label{example_distortion_non_RV_tail}
  An example of a strictly increasing $\pD\in\AMDCO{2}$ that is neither in the
  Gumbel nor in the Weibull MDA is given by
  \begin{equation}
    \pD(u) = \begin{cases}
      \frac{u}{2(1-\e^{-2})}, & u \in [0,1-\e^{-2}),\\
      1 + \frac{1}{\log(1-u)}, & u \in [1-\e^{-2},1];
    \end{cases}\label{eq_distortion_non_RV_tail}
  \end{equation}
  see Figure~\ref{fig_distortion_non_RV_tail} (right) for an illustration.
  \begin{figure}
    \begin{center}
      \includegraphics[width=0.95\textwidth]{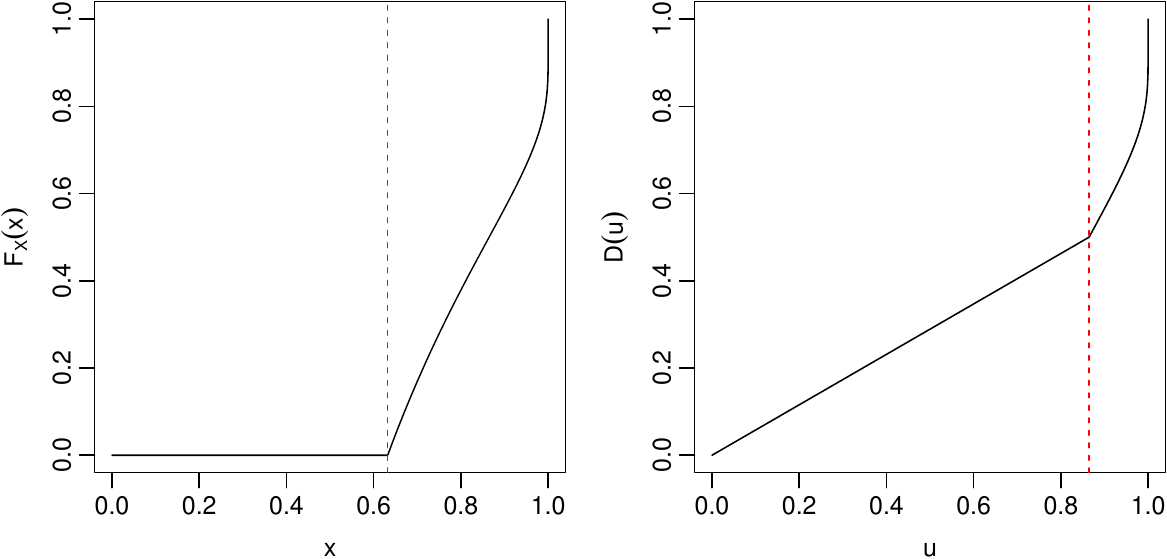}
    \end{center}
    \caption{Distribution function $\pF_X$ (left) and distortion function $\pD$
      (right) of Example~\ref{example_distortion_non_RV_tail}. The red dashed
      vertical lines indicate $1-\e^{-1}$ (left) and $1-\e^{-2}$ (right).}
    \label{fig_distortion_non_RV_tail}
  \end{figure}
  To see that there is no $\gamma>0$ such that $\pD\in\MDA(\Psi_{\gamma})$,
  denote by $X$ a random variable with distribution function
  $\pF_X(x) = 1 + 1/\log(1-x)$ for $x\geq 1-\e^{-1}$; see
  Figure~\ref{fig_distortion_non_RV_tail} (left).  Defining $Y = 1/(1-X)$, we
  then have $\pF_Y(x) = 1 - 1/\log(x)$ for $x\geq \e$, and, as noted in
  \cite[Section~2.4.4, p.~67]{BeirlantGoegebeurTeugelsSegers2004}, that
  $\pF_X\in\MDA(\Psi_{\gamma})$ if and only if $\pF_Y\in\MDA(\Phi_{\gamma})$, where $\Phi_{\gamma}$ denotes the Fréchet distribution with parameter $\gamma>0$, that is $\Phi_{\gamma}(x) = \exp(-x^{-\gamma})$ for $x\geq 0$ and zero otherwise.
  However, for the log-Pareto distribution $\pF_Y$ we have
  that $\pF_Y$ is neither in the MDA of the Fréchet nor of
  the Gumbel distribution (see \cite[Example~1.3.3, p.~14, Example~2.6.1,
  p.~96]{Galambos1987} and \cite[p.~79-80]{FalkHuslerReiss2011} for an
  additional discussion), and therefore $\pF_X$ is not in the MDA of the Weibull distribution.
  Given that $\pD$ defined in \eqref{eq_distortion_non_RV_tail} is equal to
  $\pF_X$ on $[1-\e^{-2},1]$, the same is true for $\pD$.  At the same time, $\pD$
  is continuous, strictly increasing with $\pD(0)=0$ and $\pD(1)=1$; see also
  Figure~\ref{fig_distortion_non_RV_tail} (right). The convexity of $\pD$
  finally follows from the fact that
  \begin{align*}
    \frac{\d^2}{\d x^2} \pF_X(x) = \frac{\log(1 - x) + 2}{(1 - x)^2 \log^3(1 - x)}\ge 0,\quad x \geq 1-\e^{-2},
  \end{align*}
  which implies that $\pF_X$ is convex on $[1-\e^{-2},1]$.  Furthermore, $\pD$ is constructed
  in such a way that the linear part on $[0,1-\e^{-2})$ preserves the convexity
  on all of $[0,1]$.  As such, we have $\pD\in\AMDCO{2}$. From
  Corollary~\ref{cor_d_abs_monotone_not_Gumbel} it is then clear that
  $\pD\not\in\MDA(\Lambda)$.

  At this point, it is not clear if there is a distortion $\pD\in\AMDCO{d}$,
  $d\in\{3,\dots\}\cup\{\infty\}$, that fails to be in the Weibull MDA. It is
  possible that the more restrictive $d$-absolute monotonicity for $d\geq 3$
  will force distortions into a regular variation framework, but we leave this
  point for further research.
\end{example}

A notable feature in the counter-example in
Corollary~\ref{cor_d_abs_monotone_not_Gumbel} is that the density is unbounded
at the right endpoint.  For distortion functions with non-zero, bounded
densities at the right endpoint one can indeed directly establish that they
belong to the Weibull MDA with an extreme value index equal
to one.  In case of $2$-absolutely monotone distortions the convexity then also
implies that the derivative is non-decreasing and hence non-zero at the right
endpoint, and in this case a bounded derivative at the right endpoint is hence
sufficient to apply the following proposition.

\begin{proposition}[Distortions with non-zero, bounded densities at the right
  endpoint belong to the Weibull MDA]\label{prop_D_in_Weibull_MDA}
  Let $\pD$ be an absolutely continuous distortion function such that the
  density at the right endpoint $\dD(1)=\lim_{u\nearrow 1}\dD(u)$ exists and
  $0 < \dD(1) < \infty$.  Then $\pD \in \MDA(\Psi_1)$. In particular, for
  absolutely continuous $\pD\in\AMDCO{2}$ we have that if $\dD(1) < \infty$,
  then $\pD\in\MDA(\Psi_1)$.
\end{proposition}
While the fact that $d$-absolute monotonicity restricts $\gamma$ to $(0,1]$ in Theorem~\ref{theorem_max_domain_of_attraction_Morillas} is proven using results for stdfs derived in Section~\ref{sec_trafo}, it can be reaffirmed using direct calculations
leading to a result in similar vein under additional assumptions.
\begin{example}[{Derivative based conditions for $\pD\in\MDA(\Psi_{\gamma})$, $\gamma\in(0,1]$}]
  Assume $\pD$ is a distortion function with derivatives $\pD'$ and $\pD''$ such
  that $\pD'(u) > 0$ on $(1-\varepsilon,1)$ for some
  $\varepsilon >0$. Furthermore, assume $\lim_{u\nearrow 1}(1-u)\pD'(u) = 0$ and
  \begin{align}\label{eq_condition_vonMises}
    \lim_{u\nearrow 1} \frac{(1-u)\pD''(u)}{\pD'(u)} = \kappa \in [0,1).
  \end{align}
  By L'Hôpital's rule, %
  \begin{align*}
    \lim_{u \nearrow 1} \frac{(1 - u) \pD'(u)}{1-\pD(u)} = \lim_{u \nearrow 1} \frac{-\pD'(u) + (1-u)\pD''(u)}{-\pD'(u)} = 1 - \lim_{u \nearrow 1} \frac{(1-u)\pD''(u)}{\pD'(u)} = 1-\kappa \in(0,1],
  \end{align*}
  which verifies the von Mises condition (see \cite[Corollary 3.3.13,
  p.~136]{embrechtsklueppelbergmikosch1997}) and thus establishes that
  $\pD\in\MDA(\Psi_{1-\kappa})$. In particular, for $\pD\in\AMDCO{d}$ with
  existing derivatives $\pD'\geq 0$ and $\pD''\geq 0$, the limit in
  \eqref{eq_condition_vonMises}, if it exists, cannot be negative.  As such, if
  the limit in \eqref{eq_condition_vonMises} holds, we directly have
  $\pD\in\MDA(\Psi_{\gamma})$, $\gamma\in(0,1]$, in line with Theorem~\ref{theorem_max_domain_of_attraction_Morillas}~\ref{item_distortion_theorem_2}.  As a concrete
  example, take $\pD=u\e^{\lambda(u-1)}$ as the probability generating function
  of $N = X+1$, where $X\sim\lpois(\lambda)$, $\lambda>0$.
  Proposition~\ref{prop_pgf_representation} guarantees $\pD\in\AMDCO{\infty}$,
  \begin{align*}
    \pD'(u)= \e^{\lambda(u-1)} (1+\lambda u) > 0 \quad\text{and}\quad
    \pD''(u)=\lambda \e^{\lambda(u-1)} (2+\lambda u) > 0,
  \end{align*}
  fulfilling all necessary conditions. With
  \begin{align*}
    \kappa = \lim_{u \nearrow 1} \frac{(1 - u) \pD''(u)}{\pD'(u)} &= \lim_{u \nearrow 1} \frac{(1-u)\lambda \e^{\lambda(u-1)} (2+\lambda u)}{\e^{\lambda(u-1)} (1+\lambda u)} = \lim_{u \nearrow 1} \frac{(1-u)\lambda  (2+\lambda u)}{(1+\lambda u)} = 0,
  \end{align*}
  we obtain that $\pD\in\MDA(\Psi_{1})$.
\end{example}
When considering Theorem~\ref{theorem_max_domain_of_attraction_Morillas}, specifically \ref{item_distortion_theorem_2},
it is interesting to ask in which cases we have $\gamma = 1$, leading to
$\pD^*(u) = u$ and implying that $\Copula$ and $\Copula_{\pD}$ are in the same
MDA while at the same time representing different dependence
structures (since there is no reason that $\pD$ must be the identity).  Another
interesting question concerns examples in which the full range $\gamma\in(0,1]$
can be reached. For absolute monotone distortions, we have the following first
result as a consequence of Proposition~\ref{prop_D_in_Weibull_MDA}.
\begin{corollary}[MDA for absolutely monotone distortions with finite expectation of the discrete generator]\label{corollary_abs_monotone_domain_of_attraction}
  Let $\pD\in\AMDCO{\infty}$ with an associated representation
  $\pD(u) = \sum_{n=1}^{\infty}p_n u^n$, where $p_n = \Prob(N=n)$ for a
  non-degenerate discrete random variable $N$ supported on $\IN$. If
  $\E[N]<\infty$, then $\pD\in\MDA(\Psi_1)$.
\end{corollary}
Corollary~\ref{corollary_abs_monotone_domain_of_attraction} can also be
considered through the lens of MDA for mixtures,
originally studied by \cite{berman1962limiting} and, more recently, developed in
\cite{sreehari2010extremes}. In this context, the probability generating
function $\pD(u) = \sum_{n=1}^\infty \Prob(N=n)u^n$ can be viewed as a
mixture of distribution functions $\pF_n(u)=u^n$.

Corollary~\ref{corollary_abs_monotone_domain_of_attraction} does not cover the
infinite expectation case of the associated discrete generator.
The following example shows that, in this case, values for $\gamma\in(0,1]$ are
possible. The case $\gamma < 1$ is of particular interest since, as noted in
Theorem~\ref{theorem_max_domain_of_attraction_Morillas}~\ref{item_distortion_theorem_2},
$\gamma=1$ implies that the distorted and undistorted copulas are attracted to
the same EVC. The importance of the following example is therefore to provide a
distortion where the limiting EVCs are different, thus leading to an altered
behavior of joint extreme events due to the distortion. From a practical
perspective, it is important to realize that a simulation algorithm for the
utilized zeta distribution is available, see \cite[p.~551]{Devroye1986}, and
therefore simulation from the associated distortion function
is possible via Proposition~\ref{prop:stoch:rep}.
To present our results we say that a positive measurable function $f$ is said to be \emph{regularly varying with index $\gamma$} if $\lim_{x\to\infty}f(tx)/f(x) = t^{\gamma}$, $t>0$.
In this case we write $f\in\RV_{\gamma}$; see \cite[p.~18]{BinghamGoldieTeugels1987}.
For a distribution function $\pF$ with $x_+:=\supp_+(\pF)<\infty$ regular variation of $1-\pF(x_+-1/\cdot)$ with index $-\gamma$, $\gamma>0$, is equivalent to $\pF\in\MDA(\EW_{\gamma})$; see \cite[Proposition~1.13, p.~59]{Resnick1987}.
\begin{example}[Zeta distortion]\label{example_zeta_distortion}
  This example proposes an absolutely monotone distortion function $\pD_s$, $s>1$, such that
  $\E[N]\in(0,\infty]$ for the associated discrete generator random variable
  $N$ and $\pD_s \in \MDA(\Psi_{\min\{s-1,1\}})$.
  In view of Theorem~\ref{theorem_max_domain_of_attraction_Morillas} and Corollary~\ref{cor_Pi_to_Gumbel} this allows to tune if and how the distortion impacts the MDA; see also Figure~\ref{fig_Clayton_zeta} for an example.

  For $s>1$ let $N \sim \ZetaDist(s)$, where the zeta distribution (also
  called Zipf or Lotka distribution) has the probability mass function
  $\Prob(N = n) = \frac{1}{\zeta(s)n^s}$, $n\in\IN$,
  and $\zeta(s)$ is the Riemann zeta function; see \cite[p.~550]{Devroye1986}.
  Based on $N$ we define the \emph{zeta distortion} function as the probability generating function of $N$ given by
  $\pD_s(u) = \sum_{n=1}^{\infty} \Prob(N = n) u^n = \frac{1}{\zeta(s)} \sum_{n=1}^{\infty} \frac{u^n}{n^s} = \frac{1}{\zeta(s)} \polylog_s(u)$, $u\in[0,1]$,
  where $\polylog_s(u)$ is the polylogarithm function; see \cite[\S 25.12
  (ii)]{NIST:DLMF}.
With $\lim_{s\to\infty}\polylog_s(u) = u$ and $\lim_{s\to\infty}\zeta(s) = 1$ the trivial distortion can be recovered as the limiting case $\lim_{s\to\infty}\pD_s(u)=u$.
  The zeta distribution satisfies
  $\E[N] = \frac{\zeta(s-1)}{\zeta(s)}$, which is infinite for $1 < s \leq 2$.
  For $s>2$, the regular variation of $1-\pD_s(1-1/\cdot)$ is covered by
  Corollary~\ref{corollary_abs_monotone_domain_of_attraction}, but for
  $1 <s \leq 2$, the condition $\E[N] < \infty$ in
  Corollary~\ref{corollary_abs_monotone_domain_of_attraction} is not satisfied.
  Instead, we directly establish the regular variation property of $1-\pD_s(1-1/\cdot)$
  in what follows.

  Consider $1 < s < 2$.  For the polylogarithm we have
  $\polylog_s'(u) := \frac{\d}{\d u}\polylog_s(u) =
  \frac{1}{u}\polylog_{s-1}(u)$.  Furthermore, combining \cite[(12) and (14),
  p.~30]{Bateman1953}, we have for $0 < s < 1$ that
  $\lim_{u\nearrow 1} (1-u)^{1-s}\polylog_s(u) = \Gamma(1-s)$.  A refined second
  order expansion can be found in \cite[Proposition~2]{paulsen2002behavior}.
  By l'Hôpital's rule, we have $1-\pD_s(1-1/\cdot) \in \RV_{1-s}$ for $1 < s < 2$
  since
  \begin{align*}
    \lim _{x \to \infty} \frac{1-\pD_s\left(1-\frac{1}{t x}\right)}{1-\pD_s\left(1-\frac{1}{x}\right)}&=\lim _{x \to \infty} \frac{1-\frac{1}{\zeta(s)} \polylog_s\left(1-\frac{1}{t x}\right)}{1-\frac{1}{\zeta(s)} \polylog_s\left(1-\frac{1}{x}\right)}
                                                                                                    =\lim _{x \to \infty} \frac{-\frac{1}{\zeta(s)} \polylog_s^{\prime}\left(1-\frac{1}{t x}\right) \frac{1}{t x^2}}{-\frac{1}{\zeta(s)} \mathrm{Li}^{\prime}\left(1-\frac{1}{x}\right) \frac{1}{x^2}}\\
    &=\lim _{x \to \infty} \frac{\polylog_s^{\prime}\left(1-\frac{1}{t x}\right)}{t\polylog_s^{\prime}\left(1-\frac{1}{x}\right)}
    =\lim _{x \to \infty} \frac{1}{t} \cdot \frac{1-\frac{1}{x}}{1-\frac{1}{t x}} \cdot \frac{(tx)^{2-s}\polylog_{s-1}\left(1-\frac{1}{t x}\right)}{x^{2-s}\polylog_{s-1}\left(1-\frac{1}{x}\right)} \frac{x^{2-s}}{(tx)^{2-s}}\\
                                                                                                  & =t^{-1} \cdot t^{2-s} = t^{1-s}.
  \end{align*}
  For the limiting case $s=2$, we denote by
  $\Phi_{s,\nu}(z) = \sum_{n=0}^{\infty}(n+\nu)^{-s}z^n$ Lerch’s Transcendent;
  see \cite[\S 25.14]{NIST:DLMF} and \cite[p.~27]{Bateman1953}.  Combining
  $\lim_{z\nearrow 1}\Phi_{1,\nu}(z)/(-\log(1-z)) = 1$ (see \cite[(13), p.~30]{Bateman1953})
  and $\polylog_s(z) = z\Phi_{s,1}(z)$ (see \cite[(14), p.~30]{Bateman1953}), then implies that
  \begin{align*}
    \lim _{x \to \infty} \frac{\polylog_1\left(1-\frac{1}{t x}\right)}{ \polylog_1\left(1-\frac{1}{x}\right)}
    &=\lim _{x \to \infty} \frac{\left(1-\frac{1}{t x}\right) \Phi_{1,1}\left(1-\frac{1}{t x}\right)}{\left(1-\frac{1}{x}\right) \Phi_{1,1}\left(1-\frac{1}{x}\right)} \cdot \frac{-\log \left(\frac{1}{x}\right)}{-\log \left(\frac{1}{t x}\right)} \cdot \frac{-\log \left(\frac{1}{t x}\right)}{-\log \left(\frac{1}{x}\right)}\\
    &=\lim _{x \to \infty} \frac{\left(1-\frac{1}{t x}\right) \Phi_{1,1}\left(1-\frac{1}{t x}\right)}{\left(1-\frac{1}{x}\right) -\log \left(\frac{1}{t x}\right)} \cdot \frac{-\log \left(\frac{1}{x}\right)}{\Phi_{1,1}\left(1-\frac{1}{x}\right)} \cdot \frac{\log(t) + \log(x)}{\log(x)} = 1.
  \end{align*}
  Again by l'Hôpital's rule, we have $1-\pD_s(1-1/\cdot)\in\RV_{-1}$ since
  \begin{align*}
    \lim_{x \to \infty} \frac{1-\pD_s\left(1-\frac{1}{t x}\right)}{1-\pD_s\left(1-\frac{1}{x}\right)}
    &=\lim_{x \to \infty} \frac{1-\frac{1}{\zeta(2)} \polylog_2\left(1-\frac{1}{t x}\right)}{1-\frac{1}{\zeta(2)} \polylog_2\left(1-\frac{1}{x}\right)}
      =\lim _{x \to \infty} \frac{1}{t} \frac{\polylog_1\left(1-\frac{1}{t x}\right)}{\polylog_1 \left(1-\frac{1}{x}\right)}=t^{-1}.
  \end{align*}
  Noting that we have $\E[N]<\infty$ for $s>2$ which implies, by
  Corollary~\ref{corollary_abs_monotone_domain_of_attraction}, that
  $\pD_s\in\MDA(\Psi_1)$, all cases can now be combined into
  $\pD_s \in \MDA(\Psi_{\min\{s-1,1\}})$.  This shows that as
  $\mathbb{E}[N]$ diverges, the limiting distortion function can change the tail
  behavior of a Morillas-type distorted copula as discussed in Theorem~\ref{theorem_max_domain_of_attraction_Morillas}.
\end{example}
\subsection{Intersection with the theory for Archimedean and Archimax copulas}\label{sec_Archimedean_connection}
Formally, by \cite[p.~170]{morillas2005b} (and in the bivariate case also
\cite[Corollary~7, p.~5]{durrleman2000simple} and
\cite[p.~395]{GenestRivest2001}), any Archimedean copula $\Copula_{\gen}$ is a
Morillas-type distortion of the independence copula via the distortion function
$\pD_{\gen}(u) = \gen(-\log u)$, that is $\Copula_{\gen} = \Pi_{\pD_{\gen}}$ (in
terms of our notation $C_{\pD}$ for a $\pD$-distorted copula $C$).  However, as
already mentioned after Corollary~\ref{cor_Pi_to_Gumbel}, the distortions
$\pD_{\gen}$ induced by Archimedean copulas may not satisfy the assumptions of
Theorem~\ref{theorem_max_domain_of_attraction_Morillas}
or Corollary~\ref{cor_Pi_to_Gumbel}.  The intuitive reason that $d$-absolute
monotonicity may fail for Archimedean distortions is that such distortions only
have to appropriately transform the independence copula, that is their domain is
the singleton $\{\Pi\}$. On the contrary, $d$-absolutely monotone distortions
faithfully work on the domain $\SetCopula_d$ which requires tighter
restrictions. We next discuss concrete examples of non-$d$-absolutely monotone
distortions induced by Archimedean copulas.
\begin{remark}[Morillas-type distortions linked to Archimedean copulas]
  By \cite[Theorem~2.2]{McNeilNeslehova2009},
  $\Copula_{\gen}(u_1,\dots,u_d) := \gen(\sum_{i=1}^{d}\invgen(u_i))$ is a
  $d$-dimensional copula if and only if $\gen(0) = 1$, $\lim_{t\to\infty}\gen(t)=0$
  and $\gen$ is $d$-monotone. Defining $\pD_{\gen}(u) = \gen(-\log u)$,
  see \cite[p.~170]{morillas2005b}, we have the alternative multiplicative
  representation $\Copula_{\gen} = \Pi_{\pD_{\gen}}$, that is any Archimedean
  copula is formally a Morillas-type distortion of the independence copula.
  However, the associated distortion function $\pD_{\gen}$ is \emph{not}
  necessarily $d$-absolutely monotone.

  Indeed, rephrasing \cite[Theorem~2.2]{McNeilNeslehova2009}, we obtain that for
  $\pD \colon [0,1] \to [0,1]$, $\Pi_{\pD}$ is a $d$-dimensional copula if and
  only if $\pD(1)=1$, $\lim_{u\searrow 0} \pD(u)=0$ and $\gen_{\pD}(t):=\pD(\e^{-t})$, $t\geq 0$,
  is $d$-monotone. Denoting by $\SetCMonDist$ the set of such distortions we in
  fact have $\AMDCO{d} \subsetneq \SetCMonDist$, which can be seen as follows.
  For $\pD\in\AMDCO{d}$, $\Pi_{\pD}$ is a copula and
  $\Pi_{\pD}(u_1,\dots,u_d) = \gen_{\pD}( \sum_{i}^{d}\invgen_{\pD}(u_i))$, where
  $\gen_{\pD}(t) = \pD(\e^{-t})$. However, this implies that $\gen_{\pD}$ is $d$-monotone
  and hence $\pD\in\SetCMonDist$.  To show that the inclusion is strict,
  consider the Gumbel--Hougaard copula with completely monotone generator
  $\gen_{\GH}(t) = \exp(-t^{1/\theta})$.  Then the associated distortion is
  $\pD_{\gen_{\GH}}(u) = \gen_{\GH}(-\log(u)) = \exp(-(-\log u)^{1/\theta})$, which is non-convex
  (except for the independent case $\theta=1$) and hence $\pD_{\gen_{\GH}}\not\in\AMDCO{2}$ when $\theta > 1$.

  As such, the assumption $\pD\in\AMDCO{d}$ in
  Theorem~\ref{theorem_max_domain_of_attraction_Morillas}~\ref{item_distortion_theorem_2}
  is restrictive when considering Archimedean copulas as distortions of $\Pi$.
  However, the result in \cite{GenestRivest1989} that Archimedean copulas such
  that $\lim_{t\nearrow 1}(\invgen(t)/\frac{\d}{\d t}\invgen(t))'= 1/d$ are in the
  Gumbel MDA still applies. This result is exactly in line
  with
  Theorem~\ref{theorem_max_domain_of_attraction_Morillas}~\ref{item_distortion_theorem_2}
  and Corollary~\ref{cor_Pi_to_Gumbel}, just stated in terms of our distortion
  based framework operating on the larger class $\SetCMonDist$ instead of
  $\AMDCO{d}$. The difference in assumptions stems from the fact that the
  distortions considered in
  Theorem~\ref{theorem_max_domain_of_attraction_Morillas} act on the set of all
  $d$-dimensional copulas, while distortions linked to Archimedean copulas only
  act on the independence copula $\Pi$.  This shows that considering distortions
  acting on a smaller class of copulas can reduce the necessary assumptions
  placed on the distortions.  At the same time, when this alternative class of
  distortions intersects with $\AMDCO{d}$, the obtained results necessarily have
  to match at least on the intersection.
\end{remark}
While $d$-absolute monotonicity is too strong of a requirement
for Archimedean distortions, they will generally be regularly varying. The
following lemma establishes a connection between the regular variation for
functions in case of exponential or logarithmic transformations. It is
especially interesting to apply in case $\gen$ is an Archimedean generator, in
which case the associated $\pD_{\gen}$ is a cumulative distribution function on
$[0,1]$, and the lemma hence implies $\pD_{\gen}\in\MDA(\Psi_{\rho})$ if $\gen$
is suitably regularly varying.
While the following Lemma~\ref{lemma_Archimedean_D_regular_variation} will in our context be applied to Archimedean generators and their associated distortion functions it is formulated in more general terms of non-negative measurable functions.
\begin{lemma}[Regular variation of generators and associated distortions]\label{lemma_Archimedean_D_regular_variation}
  Let $\gen \colon [0,\infty) \to [0,1]$ be measurable with $\gen(0)=1$, and $\pD \colon[0,1]\to[0,1]$ be measurable with $\pD(1)=1$.
  \begin{enumerate}
  \item If $1-\gen(1/\cdot)\in\RV_{-\rho}$,
    $\rho>0$, then $1-\pD_{\psi}(1-1/\cdot)\in\RV_{-\rho}$ for $\pD_{\psi}:=\psi\circ(-\log)$.
  \item If $1-\pD(1-1/\cdot)\in\RV_{-\gamma}$,
    $\gamma>0$, then $1-\gen_{\pD}(1/\cdot)\in\RV_{-\gamma}$ for $\gen_{\pD}:=\pD\circ(\exp^{-1})$.
  \end{enumerate}
\end{lemma}

While we so far focused on transferring properties of $\gen$ to $\pD$, the following
remark disucsses defining Archimedean copulas based on distortion functions.
\begin{remark}[Archimedean and Archimax copulas induced by $d$-absolutely monotone distortions and sampling]\label{remark_Archimedean_sampling}
  \begin{enumerate}
  \item\label{rem:AC:1} For a distortion $\pD\in\AMDCO{d}$, define a generator
    $\gen_{\pD}$ via $\gen_{\pD}(u):=\pD(\e^{-u})$.  Then $\gen_{\pD}$ is the
    $d$-monotone generator of an Archimedean copula. To see
    this note that the inverse is given as $\gen_{\pD}^{-1}(u)=-\log \qD(u)$, and we thus
    have
    \begin{align*}
      \Copula_{\gen_{\pD}}(\bfu)
      &= \pD\left(\frac{1}{\exp\left( \sum_{i=1}^{d} -\log \qD(u_i) \right)}\right)
        = \pD\left(\frac{1}{\prod_{i=1}^{d}1/\qD(u_i)}\right)
        = \pD\left(\,\prod_{i=1}^{d}\qD(u_i)\right)\\
      &= \Pi_{\pD}(\bfu),\quad \bfu\in[0,1]^d.
    \end{align*}
    Given that $\pD\in\AMDCO{d}$, we have that $\Copula_{\gen_{\pD}}=\Pi_{\pD}$
    is a copula and hence $\gen_{\pD}$ is $d$-monotone; see
    \cite[Theorem~2.2]{McNeilNeslehova2009}.  Furthermore, via
    Lemma~\ref{lemma_Archimedean_D_regular_variation}, $\gen_{\pD}$ is regularly
    varying at zero if $\pD$ is in the Weibull MDA. In this
    case we automatically have
    $\Copula_{\gen_{\pD}}\in\MDA(\GHCopula_{\theta})$ for some
  $\theta\geq 1$ via Corollary~\ref{cor_Archimedean_max_domain} below.
\item\label{rem:AC:2} Simulation of Archimedean copulas (with completely monotone generators) is
  typically done with the algorithm of \cite{marshallolkin1988}, which requires,
  as bottleneck, sampling from the underlying frailty distribution
  (corresponding to the generator by the Laplace--Stieltjes transform).
  However, in the setup of~\ref{rem:AC:1} under absolutely monotone $\pD\in\AMDCO{\infty}$,
  Proposition~\ref{prop:stoch:rep} offers an alternative approach to sampling
  of $\Copula_{\gen_{\pD}}$, which does not require to sample from
  the frailty distribution, just from the discrete generator associated
  to $\pD$.
  Circumventing the frailty distribution is a
  novel feature of this algorithm and may provide a solution to cases in which
  sampling from the frailty is challenging.
\item Combining~\ref{rem:AC:1} and \ref{rem:AC:2} allows one to draw similar
  conclusions for the more general class of Archimax copulas; see
  \cite{CharpentierFougeresGenestNeslehova2014}. First, for $\pD\in\AMDCO{d}$ and an EVC
  $\Copula\in\SetEVTCopula_d$ with associated stdf $\STDF$, we have that
  $\Copula_{\gen,\STDF}:=\Copula_{\pD}$ is an Archimax copula with an associated
  Archimedean generator $\gen_{\pD}(u):=\pD(\e^{-u})$ and stdf $\STDF$; see also \cite[p.~395]{GenestRivest2001}.
  As discussed, Archimedean copulas arise as special case for $\Copula=\Pi$.
  Moreover, for $\pD\in\AMDCO{\infty}$, Algorithm~\ref{algo_sampling_Morillas}
  can be used efficiently. This holds since for an EVC we have
  $\Copula^n(\bfu) = \Copula(\bfu^n)$ and hence~\eqref{eq:sample:bottleneck} can
  be avoided by simply sampling $(U_1,\dots,U_d)\sim\Copula$ and returning
  $\bm{W}_N = (U_1^{1/N},\dots,U_d^{1/N})$. Notably, for EVCs, the copula
  utilized for $\bm{W}_N$ is invariant under the value of $N$, while in the
  non-EVC Example~\ref{example_sampling_distorted_Clayton} the copula still
  belonged to the same functional family and only changed its parameter value
  (depending on $N$). As for purely Archimedean copulas~\ref{rem:AC:2},
  sampling based on Algorithm~\ref{algo_sampling_Morillas} avoids sampling from
  the frailty distribution which is typically necessary when sampling
  Archimax copulas. Instead, it requires sampling from the discrete generator
  linked to $\pD$.
  \end{enumerate}
\end{remark}

Based on Lemma~\ref{lemma_Archimedean_D_regular_variation} we are now also in a
position to give a short alternative proof for the fact that, under a
non-restrictive regular variation assumption, Archimedean copulas are in the
MDA of the Gumbel--Hougaard copula; see \cite[Statement
B, p.~208]{GenestRivest1989} for the original result and \cite[Proposition~2,
p.~200]{LarssonNeslehova2011} which our formulation matches.
\begin{corollary}[Archimedean copulas with regularly varying generators are in the Gumbel--Hougaard MDA]\label{cor_Archimedean_max_domain}
  Let $\Copula_{\gen}$ be an Archimedean copula with a regularly varying
  generator $\gen$ such that $1-\gen(1/\cdot)\in\RV_{-\rho}$, $0 < \rho \leq 1$.
  Then $\Copula_{\gen}\in\MDA(\GHCopula_{1/\rho})$. %
\end{corollary}
Given that any EVC lies exclusively in
  its own MDA, the preceding corollary also shows that
  an Archimedean copula with a regularly varying generator is an EVC if and only
  if it is the Gumbel--Hougaard copula. The more general proof of this result, without
  the regular variation assumption, can be found in \cite[Statement~A, p.~208]{GenestRivest1989},
  see also \cite[Theorem~4.5.2, p.~143]{Nelsen2006}, but requires the solution of a functional equation.
\section{Conclusion}\label{sec_conclusion}
In this article we substantially expanded the properties of Morillas-type
transformations. For the most common case of an absolutely monotone distortion,
we provided a stochastic representation and thus sampling algorithm, which
closes a gap in the literature and facilitates practical applications. The
stochastic representation leverages a connection between absolutely monotone
distortions and probability generating functions of discrete random variables
supported on $\IN$, and also allows one to sample from distorted joint
distributions.  Based on the maximal absolute distance between two distortions,
we quantify the maximal absolute distance between distortions of the same
underlying copula. An interesting question for future research is whether
stricter conditions, such as a small total variation distance, can lead to total
variation bounds between distorted copula measures.

Similar to the concept for copulas, we also introduced a Morillas-type
distortion for stable tail dependence functions (stdfs), which themselves
provide an elegant way of fully characterizing the dependence structure of
multivariate generalized extreme value distributions. By transforming the
associated functional equation into a domain-restricted Pexider functional
equation, we found that, in the framework of Morillas-type distorted stdfs, only
monomial distortions with an exponent less than $1$ transform \emph{any} stdf
into another stdf. Exploiting the correspondence between stdfs and extreme value
copulas (EVCs), this allowed us to characterize the subclass of distortions
which faithfully transforms EVCs into EVCs and thus to provide an alternative
proof of a key result in \cite{HerrmannHofertNeslehova2024}. Furthermore, we discussed how combining
Morillas-type transformations of stdfs with convex combinations of stdfs
provides a pathway to more flexible extremal dependence models.

We also investigated how Morillas-type distortions interact with the dependence
structure linked to the asymptotic behaviour of componentwise maxima. Under the
non-restrictive conditions of
Theorem~\ref{theorem_max_domain_of_attraction_Morillas}, the limiting EVC for
the distorted copula $\Copula_{\pD}$ is a Morillas-type distortion of the
limiting EVC of the undistorted copula $\Copula$. We provided an explicit form
of the members of this class of limiting distortions and showed that the
conditions of Theorem~\ref{theorem_max_domain_of_attraction_Morillas} are in
particular satisfied for $d$-absolutely monotone distortion functions in the
Weibull MDA. While we provided sufficient conditions which imply that the
distorted copula is in the same MDA as the initial copula (and hence features
the same limiting behaviour of componentwise extremes), we also provided a
concrete example where this is not the case. This analysis thus gave new insight
into the limiting behaviour of componentwise extremes under distorted dependence
structures.  Lastly, we connected our results to the well established theory of
Archimedean and Archimax copulas, where, if applicable, our sampling algorithm has the novel
feature of circumventing the sampling from the associated frailty distribution
that is otherwise typically necessary.

\subsection*{Acknowledgments}
The first and fourth authors gratefully acknowledge support from the Natural Sciences and Engineering Research Council (NSERC) of Canada via NSERC grant RGPIN-2020-05784.
The first author also gratefully acknowledges administrative support of the Karlsruhe Institute of Technology (KIT) where part of this research was carried out.

\printbibliography[heading=bibintoc]

\appendix

\section{Auxiliary results}
\begin{lemma}[Convergence of function composition] \label{composition_convergence}
	Denote by $(Y,\d_Y)$ and $(Z,\d_Z)$ two metric spaces.
	For a set $X$ let $g_n \colon X \to Y$ and $f_n \colon Y \to Z$ be two sequences of functions.
	If $g_n$ converges pointwise to a function $g \colon X \to Y$, and $f_n$ converges uniformly to a continuous function $f \colon Y \to Z$ then $f_n\circ g_n$ converges pointwise to $f\circ g$.
\end{lemma}
\begin{proof}
	For $x\in X$ denote by $y_n = g_n(x)$ and $y = g(x)$.
	Then via the triangle inequality we have
	\begin{align*}
		\d_Z ( f_n(y_n),\, f(y)) \leq \d_Z(f_n(y_n),\, f(y_n)) + \d_Z(f(y_n),\, f(y)).
	\end{align*}
	Uniform convergence of $f_n$ to $f$, that is for every $\varepsilon > 0$ there is a $N_1\in\IN$ such that $\sup_{y\in Y}\d_Z(f_n(y),$ $f(y)) \leq \varepsilon/2$ when $n\geq N_1$, implies that the first term on the right hand side is less than $\varepsilon/2$ for $n\geq N_1$.
	The convergence $y_n \to y$ implies that for any $\delta > 0$ there is $N_2\in\IN$ such that $\d_Y(y_n,y) < \delta$ for $n\geq N_2$.
	Combining this with the continuity of $f$, i.e, for any $\varepsilon > 0$ there is a $\delta > 0$ such that $\d_Y(y,y') < \delta$ implies $\d_Z(f(y),f(y')) < \varepsilon$, we can find for every $\varepsilon/2$ a $N_2\in\IN$ such that $\d_Z(f(y_n),f(y)) < \varepsilon/2$ for $n\geq N_2$.
	For $n \geq N=\max(N_1,N_2)$ we then have $\d_Z ( f_n(y_n),\, f(y)) < \varepsilon/2+\varepsilon/2 = \varepsilon$.
\end{proof}
\begin{example}[Kumaraswamy distortion function in $\MDA(\EW_\gamma)$ but not $d$-absolutely monotone]\label{Kumaraswamy_distortion}
	Consider the distortion function
	\begin{align*}
	K_{a,\gamma}(x) = 1 - (1 - x^a)^\gamma, \quad x \in [0,1],
\end{align*}
which is the cumulative distribution function of the Kumaraswamy distribution
with parameters $a > 0$ and $\gamma > 0$. Using the von Mises Condition in
\cite[Corollary~3.3.13, p.~136]{embrechtsklueppelbergmikosch1997}, l'Hôpital's rule
implies an upper support endpoint of $x_{K_{a,\gamma}}=1$ that
		\begin{align*}
			\lim_{x \nearrow 1} \frac{(1 - x) K_{a,\gamma}'(x)}{1-K_{a,\gamma}(x)}
                  &= \lim_{x \nearrow 1} \frac{(1 - x)\, a\gamma\, x^{a-1}(1 - x^a)^{\gamma - 1}}{(1 - x^a)^\gamma} = \lim_{x \nearrow 1} a\gamma\, x^{a - 1} \cdot \frac{1 - x}{1 - x^a}\\
                  &= \lim_{x \nearrow 1} a\gamma\, x^{a - 1} \cdot \frac{1}{a x^{a - 1}} = \gamma > 0.
		\end{align*}
		This matches the von Mises condition for
                the MDA of the Weibull distribution
                with index $\gamma$. Thus,
                $K_{a,\gamma} \in \mathrm{MDA}(\Psi_\gamma)$.

		However, $K_{a,\gamma}$ is not $d$-absolutely monotone.
		While the first derivative of $K_{a,\gamma}$ is
		\begin{align*}
		K_{a,\gamma}'(u) = a\gamma\, u^{a-1} (1 - u^a)^{\gamma-1} > 0 \quad \text{for } u \in (0,1),
		\end{align*}
		the second derivative is
		\begin{align*}
			K_{a,\gamma}''(u) = a\gamma \left[ (a - 1)u^{a - 2}(1 - u^a)^{\gamma - 1} - a(\gamma - 1)u^{2a - 2}(1 - u^a)^{\gamma - 2} \right],
		\end{align*}
		which is not necessarily non-negative on $(0,1)$. For example, when $a = \gamma = 2$, this becomes
		\begin{align*}
			K_{a,\gamma}''(u) = 4 \left[ u^0 (1 - u^2)^1 - 2u^2 \right] = 4[(1 - u^2) - 2u^2] = 4(1 - 3u^2),
		\end{align*}
		which is negative for $u > 1/\sqrt{3}$. Therefore, $K_{a,\gamma}$ is not $2$-absolutely monotone or $d$-absolutely monotone of any order on $[0,1]$.
\end{example}

\section{Proofs}
\begin{proof}[Proof of Proposition~\ref{prop_distorted_measure_difference}]
  \mbox{}
  \begin{enumerate}
  \item %
    Define $\tilde{\bfu}=(\tilde{u}_1,\dots,\tilde{u}_d)$ for $\bfu=(u_1,\dots,u_d)\in[0,1]^d$ via $\tilde{u}_i:=\pD'\left(\qD(u_i)\right)$ for $i\in\{1,\dots,d\}$. Since $\pD$ is bijective, we have
    \begin{align*}
      \sup_{u \in [0,1]} \vert \pD'\circ\qD(u) - u \vert
      =\sup_{u \in [0,1]} \vert \pD'(u) - \pD(u) \vert
      =\Vert \pD'-\pD \Vert_\infty
      \leq \varepsilon.
    \end{align*}
    Lipschitz continuity of copulas applied to $\Copula_{\pD'}$, see \cite[Theorem 2.10.7, p.~46]{Nelsen2006},
    implies, together with the fact that $\pD'$ is bijective, that
    \begin{align*}
      \abs{ \Copula_{\pD'}(\bfu) - \Copula_{\pD}(\bfu) }
      &= \abs{ \Copula_{\pD'}(\bfu) -  \Copula_{\pD'}(\tilde{\bfu}) + \Copula_{\pD'}(\tilde{\bfu})- \Copula_{\pD}(\bfu) }\\
      &\leq \abs{ \Copula_{\pD'}(\bfu) - \Copula_{\pD'}(\tilde{\bfu}) } + \abs{ \Copula_{\pD'}(\tilde{\bfu}) - \Copula_{\pD}(\bfu) }\\
      &\leq \Vert \bfu - \tilde{\bfu}\Vert_1 + \abs{ \pD'\bigl(\Copula(\qD(\bfu))\bigr) - \pD\bigl(\Copula(\qD(\bfu))\bigr)}\\
      &\leq \sum_{i=1}^d \abs{u_i - \pD'\circ\qD(u_i)} + \varepsilon \leq d \varepsilon + \varepsilon=(d+1)\varepsilon,
    \end{align*}
    hence $\sup_{\bfu \in [0,1]^d} \abs{\Copula_{\pD'}(\bfu) - \Copula_{\pD}(\bfu)} \leq (d+1)\varepsilon$.
  \item The $\Copula$-volume of a set
    $(\bfa, \bfb]= \prod_{i=1}^d(a_i,b_i] \subseteq [0,1]^d$ is given by (the inclusion-exclusion formula)
    \begin{align*}
      \vol_{(\bfa, \bfb]}\Copula = \sum_{\bm{i}\in \{0,1\}^d} (-1)^{\sum_{j=1}^d i_j} \Copula(\bm{u}^{(\bm{i})}),
    \end{align*}
    where $\bfu^{(\bm{i})} = (a_1^{i_1}b_1^{1-i_1},\dots,a_d^{i_d}b_d^{1-i_d})$, $\bm{i}\in\{0,1\}^d$.
    The first part of the proof implies for $\Copula_{\pD'}$ and $\Copula_{\pD}$ that
    \begin{align*}
      \abs{\vol_{(\bfa, \bfb]}\Copula_{\pD'} - \vol_{(\bfa, \bfb]}\Copula_{\pD}}
      &= \biggl|\sum_{\bm{i} \in \{0,1\}^d} (-1)^{\sum_{j=1}^d i_j} \bigl(\Copula_{\pD'}(\bfu^{(\bm{i})}) - \Copula_{\pD}(\bfu^{(\bm{i})}) \bigr)\biggr| \\
      &\leq \sum_{\bm{i} \in \{0,1\}^d} \abs{ \Copula_{\pD'}(\bfu^{(\bm{i})}) - \Copula_{\pD}(\bfu^{(\bm{i})})} \leq \sum_{\bm{i} \in \{0,1\}^d}  (d+1) \varepsilon = 2^d  (d+1) \varepsilon,
    \end{align*}
    which is independent of $\bfa$ and $\bfb$. We thus obtain the uniform bound
    \begin{align*}
      \sup_{(\bfa,\bfb]\subseteq[0,1]^d} \abs{\vol_{(\bfa,\bfb]}\Copula_{\pD'} - \vol_{(\bfa,\bfb]}\Copula_{\pD}} &\leq 2^d (d+1) \varepsilon. \qedhere
    \end{align*}
  \end{enumerate}
\end{proof}

\begin{proof}[Proof of Corollary~\ref{cor_convergence_C_Dn_to_C_D}]
  Fix $\varepsilon > 0$. Since $\pD_n$ is non-decreasing and $\pD$ is
  continuous, the convergence of $\pD_n$ to $\pD$ is in fact uniform; see
  \cite[p.~1]{Resnick1987}. There thus exists $N \in \IN$ such that, for all $n \geq N$, we have
  $\Vert \pD_n-\pD \Vert_{\infty} \leq \varepsilon/(d+1)$.  Given that
  $\Copula_{\pD}$ is a copula,
  Proposition~\ref{prop_distorted_measure_difference}~\ref{prop_max_difference_I}
  implies that $\Vert \Copula_{\pD_n}-\Copula_{\pD} \Vert_{\infty} \leq \varepsilon$ for
  every $n\geq N$, even if $\Copula_{\pD_n}$ fails to be a copula. This shows the
  uniform convergence of $(\Copula_{\pD_n})$ to $\Copula_{\pD}$.
\end{proof}

\begin{proof}[Proof of Proposition~\ref{prop_pgf_representation}]
  For necessity, if $\pD$ is an absolutely monotone distortion function, then
  \cite[Theorem~2, p.~223]{feller1971} implies that $\pD$ can be written for
  $u\in[0,1)$ in the form $\pD(u) = \sum_{n=0}^{\infty} p_n u^n$ for constants
  $p_n \geq 0$, $n\in\IN_0$. By assumption we also have
  $\lim_{u\nearrow 1}\pD(u) = \pD(1) = 1$, %
  which is now equivalent to $\lim_{u\nearrow 1}\sum_{n=0}^{\infty} p_n u^n = 1$.  It follows that
  $\sum_{n=0}^{\infty} p_n = 1$, that is $p_n = \Prob(N=n)$ for some discrete
  random variable $N$ taking values in $\IN_0$. Since, by assumption,
  we also have $0=\pD(0)$ %
  and $\pD(0)=p_0$ from the series representation, we have $p_0=0$
  and thus the support of $N$ is indeed $\IN$.

  For sufficiency, if $\pD(u) = \sum_{n=1}^{\infty} \Prob(N=n)u^n$, $u\in[0,1]$,
  for some discrete random variable $N$ with support $\IN$, then
  \cite[Theorem~2, p.~223]{feller1971} implies that $\pD$ is absolutely monotone
  on $[0,1)$. The series representation also trivially implies that
  $\pD(0)=0$. Lastly, according to Abel's continuity theorem, see \cite[8.2~Theorem,
  p.~174]{Rudin1976}, $\lim_{u\nearrow 1} \pD(u) = \lim_{u\nearrow 1} \sum_{n=1}^{\infty}
  \Prob(N=n)u^n = \sum_{n=1}^{\infty} \Prob(N=n) = 1$, which establishes the
  continuity of $\pD$ at $1$ and thus also implies that $\pD(1)=1$.  This makes
  $\pD$ a distortion function.
\end{proof}

\begin{proof}[Proof of Corollary~\ref{cor_bounded_density}]
  Taking the derivative of the power series $\pD$ we have
  $\dD(u) = \sum_{n=1}^{\infty} n\Prob(N=n)u^{n-1}$, with boundary
  values $\dD(0)=\Prob(N=1)$ and $\dD(1)=\E[N]$ (possibly $\infty$).
  For necessity, if $\dD$ is bounded by $C>0$, then
  $\dD(1) =\E[N]<C$.  For sufficiency, if
  $\E[N]=\sum_{n=1}^{\infty} n\Prob(N=n)<C<\infty$, Abel's continuity
  theorem of \cite[8.2 Theorem, p.~174]{Rudin1976} implies that
  $\lim_{u\nearrow 1}\dD(u) = \E[N] < \infty$, so $\dD$ is continuous at $1$
  (and trivially on $[0,1)$), so $\dD$ is continuous on $[0,1]$ and thus
  bounded there.
\end{proof}

\begin{proof}[Proof of Proposition~\ref{prop:stoch:rep}]
  The representation $\pD(u)=\sum_{n=1}^\infty p_n u^n$ for
  $\pD\in\AMDCO{\infty}$ in Proposition~\ref{prop_pgf_representation} implies
  that $\pD$ is strictly increasing and continuous, and so is $\qD$.  Since
  $\qD$ and $\pD$ are strictly increasing,
  \cite[Proposition~1\,(3)]{embrechtshofert2013c} and conditioning on $N$
  implies that, for all $\bm{u}\in[0,1]^d$, we have
  \begin{align*}
    \Prob(\bm{U}_{\pD}\le\bm{u})
    &=\Prob\bigl(\pD(W_{N,1})\le u_1,\dots,\pD(W_{N,d})\leq u_d\bigr)\\
    &=\Prob\bigl(W_{N,1}\le \qD(u_1),\dots,W_{N,d}\leq\qD(u_d)\bigr)\\
    &=\sum_{n=1}^\infty p_n\Prob\bigl(W_{N,1}\le \qD(u_1),\dots,W_{N,d}\leq\qD(u_d)\,\big|\,N=n\bigr)\\
    &=\sum_{n=1}^\infty p_n\Copula\bigl(\qD(u_1),\dots,\qD(u_d)\bigr)^n\\
    &=\pD\bigl(\Copula\bigl(\qD(u_1),\dots,\qD(u_d)\bigr)\bigr) = \Copula_{\pD}(\bfu).
  \end{align*}
  The stochastic representation in~\eqref{eq:sample:bottleneck} holds since, for
  $\bfU_1,\dots,\bfU_n\isim \Copula$, one has
  \begin{align*}
    \Prob\bigl(\max_{1\le i\le n}\{U_{i,1}\}\le u_1,\dots,\max_{1\le i\le n}\{U_{i,d}\}\le u_d\bigr)
    &= \prod_{i=1}^n\Prob\bigl(U_{i,1}\le u_1,\dots,U_{i,d}\le u_d\bigr) = \Copula^n(\bfu). \qedhere
  \end{align*}
\end{proof}

\begin{proof}[Proof of Lemma~\ref{lem:prop:stdf}]
  \mbox{}
  \begin{enumerate}
  \item %
    For all $j=1,\dots,d$, $\max_{1\le j\le d}\{|x_j|W_j\}\ge |x_j|W_j$ almost
    surely and thus $\STDF(\bm{x})=\lVert\bm{x}\rVert_{\text{D}}\ge
    |x_j|$. Therefore, $\STDF(\bm{x})\ge\max\{x_1,\dots,x_d\}$,
    $\bm{x}\in[0,\infty)^d$. Furthermore,
    $\max_{1\le j\le d}\{|x_j|W_j\}\le \sum_{j=1}^d|x_j|W_j$ almost surely and
    thus $\STDF(\bm{x})=\lVert\bm{x}\rVert_{\text{D}}\le
    \sum_{j=1}^d|x_j|$. Therefore, $\STDF(\bm{x})\le\sum_{j=1}^dx_j$,
    $\bm{x}\in[0,\infty)^d$. The remaining part of the statement is easy to
    verify via~\eqref{eq:repr:EVC}.
  \item The first two statements are a direct consequence
    of~\ref{lem:prop:stdf:1}. To see that $\STDF(\bm{1})$ can attain any value in
    $[1,d]$ for a suitable stdf $\STDF$, consider the comonotone $\STDF$ to
    obtain $\STDF(\bm{1})=1$ and the logistic stdf
    $\STDF_{\eta}(\bfx)=\left( \sum_{i=1}^{d}\abs{x_i}^{1/\eta} \right)^{\eta}$
    with parameter $\eta\in(0,1]$ for all other values $(1,d]$ since
    $\STDF_{\eta}(\bm{1}) = \e^{\eta \log(d)}=d^{\eta}$ in this case, which, as a
    function of $\eta$, can attain all values in $(1,d]$. \qedhere
  \end{enumerate}
\end{proof}

\begin{proof}[Proof of Theorem~\ref{thm:power:distortions:of:stdfs}]
  Consider sufficiency. %
  If $g(x) = cx^{\gamma}$ for $\gamma\in(0,1]$ and
  $c\in(0,\infty)$, then $g^{-1}(x)=(x/c)^{1/\gamma}$ and homogeneity of degree
  $1$ of $\STDF$ implies that
  \begin{align*}
    \STDF_g(\bm{x}) &= c\STDF\bigl((x_1/c)^{1/\gamma},\ldots,(x_d/c)^{1/\gamma}\bigr)^{\gamma} = \bigl(c^{1/\gamma}\STDF\bigl((x_1/c)^{1/\gamma},\ldots,(x_d/c)^{1/\gamma}\bigr)\bigr)^{\gamma}\\
    &=\bigl(\STDF\bigl(x_1^{1/\gamma},\ldots,x_d^{1/\gamma}\bigr) \bigr)^{\gamma},
  \end{align*}
  which is a stdf; see \cite[Proposition~2.6.1, p.~132]{Falk2019} and \cite[Lemma~2.3\,(ii)]{ChatelainFougeresNeslehova2020}.

  Consider necessity. %
  From the stdf $\STDF(\bm{x}) = \sum_{i=1}^d x_i$ of the independence copula,
  we obtain that
  $1\stackrel{!}{=} \STDF_{g}(\bm{e}_1) = g(g^{-1}(1) + (d-1)g^{-1}(0))$.
  Applying $g^{-1}$ to both sides of the equation yields $g^{-1}(0)=0$. From
  here we see that, by strict increasingness of $g^{-1}$, that $g^{-1}(1)>0$,
  and also that $g(0)=0$, which already fixes any possible solution $g$ at
  zero. Using that $\STDF_{g}$ needs to be homogeneous of degree $1$, we can use
  that $\STDF$ is homogeneous of degree $1$ and that $g^{-1}(1)>0$ to arrive, for
  all $t>0$, at the condition
  \begin{align*}
    g\bigl(g^{-1}(t)\STDF(\bm{1})\bigr)&=g\bigl(\STDF(g^{-1}(t),\ldots,g^{-1}(t))\bigr)=\STDF_{g}(t,\ldots,t)\stackrel{!}{=} t \STDF_{g}(\bm{1})\\
                                           &= tg\bigl(\STDF\bigl(g^{-1}(1),\ldots,g^{-1}(1)\bigr)\bigr) = g(g^{-1}(t)) g\bigl(\STDF\bigl(g^{-1}(1),\ldots,g^{-1}(1)\bigr)\bigr)\\
    &= g(g^{-1}(t)) g\bigl(g^{-1}(1)\STDF(\bm{1})\bigr).
  \end{align*}
  Let $x:=g^{-1}(t)$ and $y:=\STDF\left(\bm{1}\right)$.  Then
  $x\in[0,\infty)$ and, by Lemma~\ref{lem:prop:stdf}~\ref{lem:prop:stdf:2}, $y\in [1,d]$, with all values
  attainable by a suitable $\STDF$. In terms of $x,y$ we thus obtain the
  functional equation
  \begin{align}
    g(xy)=g(x)g(g^{-1}(1)y),\quad x\in[0,\infty),\ y\in[1,d].\label{eq:fct:eq}
  \end{align}
  This is the multiplicative form of a (domain-)restricted Pexider functional
  equation of the form $h(xy)=f(x)\tilde{g}(y)$ for
  $(x,y)\in R:=(0,\infty)\times(1,d)$.  By \cite[Corollary~5, p.~81]{aczel1987},
  the continuous solutions are of the form
  \begin{align*}
	  h(t)&=bc t^{\gamma},\quad t\in(0,\infty),\\
	  f(t)&=ct^{\gamma},\quad t\in(0,\infty),\\
	  \tilde{g}(t)&=bt^{\gamma},\quad t\in(1,d),
  \end{align*}
  for $b\neq 0$, $c\neq 0$ and $\gamma\in\IR$. Comparing the solutions to~\eqref{eq:fct:eq},
  we see that $h$ and $f$ must be $g$, so $b\stackrel{!}{=}1$ and therefore that
  $g(t)\stackrel{!}{=}ct^{\gamma}$ for $t > 0$. This implies that
  $g^{-1}(t)=(t/c)^{1/\gamma}$. The equation for $\tilde{g}$ is then
  automatically fulfilled since
  $\tilde{g}(t)=g(g^{-1}(1)t)=c((1/c)^{1/\gamma}t)^\gamma=t^\gamma$.

  It remains to consider the admissible values for the parameters $c\neq 0$ and
  $\gamma\in\IR$.
  From Definition~\ref{def_stdf},
  any stdf $\STDF$ is necessarily
  non-negative, so $c>0$. For the range of $\gamma$, first consider
  $\gamma \leq 0$. As $c>0$, $g$ is then either constant or strictly decreasing,
  which contradicts the assumption of $g$ being strictly increasing.  Now
  consider $\gamma>1$. With the stdf
  $\STDF(\bm{x}) = \sum_{j=1}^{d} x_j$ of the independence copula, we
  then have $\STDF_g(\bm{1}) = d^{\gamma} > d$ independent of $c>0$, which
  violates Lemma~\ref{lem:prop:stdf}~\ref{lem:prop:stdf:2}.  As such, we necessarily have
  $\gamma\in(0,1]$. As last step, note that the previously identified
  values $g(0)=g^{-1}(0)=0$ allow to extend the domain of the identified
  function $g$ (and its inverse $g^{-1}$) from $(0,\infty)$ to $[0,\infty)$, in
  line with the necessary continuity and strict monotonicity of $g$.
\end{proof}

\begin{proof}[Proof of Corollary~\ref{cor_uniqueness_up_to_scaling}]
  For $c>0$, $g_c(x):=g(cx)$ also satisfies the assumptions of
  Theorem~\ref{thm:power:distortions:of:stdfs}, with corresponding inverse
  $g_c^{-1}(x)=g^{-1}(x)/c$. Homogeneity of degree $1$ of the underlying stdf
  $\STDF$ then implies that
  $\STDF_{g_c}(\bm{x})=g_c\bigl(\STDF(g_c^{-1}(x_1),\ldots,g_c^{-1}(x_d))\bigr)=g\bigl(\STDF(g^{-1}(x_1),\ldots,g^{-1}(x_d))\bigr)=\STDF_g(\bm{x})$, $\bm{x}\in[0,\infty)^d$, so that $g_c$ generates the same $\STDF_g$ as $g$.
\end{proof}

\begin{proof}[Proof of Corollary~\ref{cor_EVC_to_EVC}]
  Consider sufficiency. If
  $\pD(u) = \pD_{\lambda,\gamma}(u) = \e^{-\lambda(-\log u)^\gamma}$,
  $u\in[0,1]$, for $\lambda>0$ and $\gamma\in(0,1]$, then
  \cite[Theorem~3.6]{morillas2005b}, see also \cite[Theorem~5]{Ressel2012},
  guarantees that $\Copula_{\pD}$ is a $d$-dimensional copula since, by the
  proof of \cite[Corollary~A.3]{HerrmannHofertNeslehova2024},
  $\pD_{\lambda,\gamma}$ is absolutely montone of any order. Furthermore, it is
  easily seen that
  $\pD_{\lambda,\gamma}(u^t) = \left(\pD(u)\right)^{t^{\gamma}}$ for any $t>0$.  With
  $\pD_{\lambda,\gamma}^{-1} = \pD_{\lambda^{-1/\gamma},\gamma^{-1}}$ and the fact that
  $\Copula$ is extreme value, we obtain for any $t>0$ that
  \begin{align*}
    \Copula_{\pD}(\bm{u}^t) &= \pD_{\lambda,\gamma} \left( \Copula \left( \pD_{\lambda^{-1/\gamma},\gamma^{-1}}(u_1^t),\dots,\pD_{\lambda^{-1/\gamma},\gamma^{-1}}(u_d^t)\right)\right)\\
    &= \pD_{\lambda,\gamma} \left( \Copula \left( \left(\pD_{\lambda^{-1/\gamma},\gamma^{-1}}(u_1)\right)^{t^{1/\gamma}},\dots,\left(\pD_{\lambda^{-1/\gamma},\gamma^{-1}}(u_d)\right)^{t^{1/\gamma}}\right)\right)\\
                                      &= \pD_{\lambda,\gamma} \left( \Copula^{t^{1/\gamma}} \left( \pD_{\lambda^{-1/\gamma},\gamma^{-1}}(u_1),\dots,\pD_{\lambda^{-1/\gamma},\gamma^{-1}}(u_d) \right)\right)\\
    &= \pD_{\lambda,\gamma}^t \left( \Copula \left( \pD_{\lambda,\gamma}^{-1}(u_1),\dots,\pD_{\lambda,\gamma}^{-1}(u_d) \right)\right)= \Copula_{\pD}^t(\bm{u}),\quad\bm{u}\in[0,1]^d,
  \end{align*}
  which demonstrates that $\Copula_{\pD}$ is max-stable and thus an EVC.

  Consider necessity. If, for any EVC $\Copula(\bm{u}) = \e^{-\STDF(-\log u_1,\dots,-\log u_d)}$ with stdf $\ell$,
  $\Copula_{\pD}$ is an EVC, there is a stdf
  $\STDF_{\pD}$ associated to $\Copula_{\pD}$ such
  that
  \begin{align}
    \Copula_{\pD}(\bm{u}) &= \e^{-\STDF_{\pD}(-\log u_1,\dots,-\log u_d)}.\label{eq_Morillas1}
  \end{align}
  By~\eqref{eq_MorillasDef}, $\Copula_{\pD}$ can alternatively be written in terms of $\pD$ and $\Copula$ via
  \begin{align}\label{eq_Morillas2}
    \Copula_{\pD}(\bm{u}) &= \pD\bigl(\e^{-\STDF(-\log \qD(u_1),\dots,-\log \qD(u_d))}\bigr).
  \end{align}
  Equating \eqref{eq_Morillas1} and \eqref{eq_Morillas2}, we must have
  \begin{align*}
    \STDF_{\pD}(\bm{x}) = -\log \pD\bigl(\e^{-\STDF\left(-\log \qD(\e^{-x_1}),\dots,-\log \qD(\e^{-x_d})\right)}\bigr),\quad\bm{x}\in[0,\infty)^d,
  \end{align*}
  that is $\STDF_{\pD}(\bm{x}) = g\bigl(\STDF(g^{-1}(x_1),\dots,g^{-1}(x_d))\bigr)=\STDF_g(\bm{x})$ for
  $g(x) = -\log \pD(\e^{-x})$, with corresponding inverse
  $g^{-1}(x) = -\log \qD(\e^{-x})$. As $\Copula$, and hence
  $\STDF$, were arbitrary, this means that $\STDF \mapsto \STDF_g$ transforms any
  stdf $\STDF$ into a stdf
  $\STDF_g$. Since $\pD$ is continuous and strictly increasing, the same is true
  for $g$. Hence we can apply Theorem~\ref{thm:power:distortions:of:stdfs} with
  $g(x) = -\log \pD(\e^{-x})$ and obtain that $g(x) = \lambda x^{\gamma}$ for
  some constants $\lambda>0$ and $\gamma\in(0,1]$.  From
  $-\log \pD(\e^{-x}) = \lambda x^{\gamma}$ it then follows that
  $\pD(u) = \e^{-\lambda(-\log u)^{\gamma}}$.
\end{proof}

\begin{proof}[Proof of Lemma~\ref{rate_function_convergence}]
  \mbox{}
  \begin{itemize}[labelwidth=\widthof{1) $\Rightarrow$ 2):}]
  \item[\ref{rate_function_convergence_I} $\Rightarrow$ \ref{rate_function_convergence_II}:]
    We will first show that
    \begin{align}\label{eq_rate_lemma}
      \abs{ C^{\rate_n}\bigl(u_1^{1/\rate_n},\ldots,u_d^{1/\rate_n}\bigr) - C^{\ceil{\rate_n}}\bigl(u_1^{1/\ceil{\rate_n}}, \dots, u_d^{1/\ceil{\rate_n}}\bigr) } \to 0,
    \end{align}
    where $\ceil{\rate_n}$ is the smallest integer at least as large as $\rate_n$.
    Given that $\ceil{\rate_n}$ is an integer-valued sequence diverging to
    infinity, the assumption $\Copula\in\MDA(\Copula^*)$ implies that
    \begin{align*}
      C^{\ceil{\rate_n}}\bigl(u_1^{1/\ceil{\rate_n}}, \dots,u_d^{1/\ceil{\rate_n}}\bigr) \to C^*(u_1, \dots u_d).
    \end{align*}
    By adding and subtracting
    $\Copula^{\ceil{\rate_n}}\bigl( u_1^{1/\rate_n}, \dots, u_d^{1/\rate_n}\bigr)$ inside
    the absolute value in \eqref{eq_rate_lemma}, the triangle inequality implies
    that
    \begin{align*}
      \bigl|  \Copula^{\rate_n}\bigl(u_1^{1/\rate_n}, \dots, u_d^{1/\rate_n}\bigr)-\Copula^{\ceil{\rate_n}}\bigl(u_1^{1/\ceil{\rate_n}}, \dots, u_d^{1/\ceil{\rate_n}}\bigr) \bigr| &\leq A_1 + A_2,
    \end{align*}
    where
    \begin{align*}
      A_1 &= \big| \Copula^{\ceil{\rate_n}}\bigl( u_1^{1/\rate_n}, \dots, u_d^{1/\rate_n} \bigr) - \Copula^{\rate_n}\bigl( u_1^{1/\rate_n}, \dots, u_d^{1/\rate_n} \bigr) \bigr|,\\
      A_2 &= \big| \Copula^{\ceil{\rate_n}}\bigl(u_1^{1/\ceil{\rate_n}}, \dots, u_d^{1/\ceil{\rate_n}}\bigr) - \Copula^{\ceil{\rate_n}}\bigl( u_1^{1/\rate_n}, \dots, u_d^{1/\rate_n} \bigr) \bigr|.
    \end{align*}
    We have left to show that each of $A_1$, $A_2$ converges to $0$ as
    $n \to \infty$.  For $A_1$, the Fr\'echet--Hoeffding bounds theorem (see
    \cite{frechet1935}, \cite{hoeffding1940}, \cite{frechet1951})
    and the fact that $u_j^{1/\rate_n} \to 1$ as $n \to \infty$
    imply that
    \begin{align*}
      A_1 & =  \big|\Copula^{\rate_n+\ceil{\rate_n} -\rate_n }\big( u_1^{1/\rate_n}, \dots, u_d^{1/\rate_n} \bigr)-\Copula^{\rate_n}\big( u_1^{1/\rate_n}, \dots, u_d^{1/\rate_n} \bigr)\bigr| \\
          & =  \Copula^{\rate_n}\big( u_1^{1/\rate_n}, \dots, u_d^{1/\rate_n} \bigr) \big| 1- \Copula^{\ceil{\rate_n} -\rate_n }\big( u_1^{1/\rate_n}, \dots, u_d^{1/\rate_n} \bigr) \bigr| \\
          & \leq \min (u_1,\dots, u_d) \big( 1-W\big(u_1^{1/\rate_n}, \dots, u_d^{1/\rate_n}\bigr)\bigr) \\
          & =\min (u_1,\dots, u_d) \biggl( 1- \max\biggl( 1-d+\sum_{j=1}^d u_j^{1/\rate_n}, 0 \biggr)\biggr) \to 0.
    \end{align*}
    For $A_2$, $\Copula^{\ceil{\rate_n}}\bigl( u_1^{1/\ceil{\rate_n}}, \dots,
    u_d^{1/\ceil{\rate_n}}\bigr)$ is a Morillas-type distorted $\Copula$
    and thus a copula. Therefore, by Lipschitz continuity of copulas, see \cite[Theorem~2.10.7, p.~46]{Nelsen2006}, we have
    \begin{align*}
      A_2 & = \big| \Copula^{\ceil{\rate_n}}\big(u_1^{1/\ceil{\rate_n}}, \dots, u_d^{\ceil{\rate_n}}\bigr) - \Copula^{\ceil{\rate_n}}\big( u_1^{\ceil{\rate_n}/(\rate_n \ceil{\rate_n})}, \dots, u_d^{\ceil{\rate_n}/(\rate_n \ceil{\rate_n})} \bigr) \bigr| \\ & \leq \sum_{i=1}^{d} \big| u_i - u_i^{\ceil{\rate_n}/\rate_n}\bigr| \; \to 0,
    \end{align*}
    where we used that $\lim_{n\to \infty} \ceil{\rate_n}/\rate_n= 1$ since the
    fraction is bounded by $1 \leq \ceil{\rate_n}/\rate_n \leq (\rate_n+1)/\rate_n$.
    Hence we have proven~\eqref{eq_rate_lemma} and can continue
    with the main argument. By the triangle inequality,
    \begin{align*}
      \bigl| C^{\rate_n}(u_1^{1/\rate_n},\ldots,u_d^{1/\rate_n}) - \Copula^*(u_1,\dots,u_d) \bigr| \leq A_3+A_4,
    \end{align*}
    where
    \begin{align*}
      A_3&=\abs{ C^{\rate_n}\bigl(u_1^{1/\rate_n},\ldots,u_d^{1/\rate_n}\bigr) - C^{\ceil{\rate_n}}\bigl(u_1^{1/\ceil{\rate_n}},\dots,u_d^{1/\ceil{\rate_n}}\bigr) },\\
      A_4&=\abs{ C^{\ceil{\rate_n}}\bigl(u_1^{1/\ceil{\rate_n}},\dots,u_d^{1/\ceil{\rate_n}}\bigr) - \Copula^*(u_1,\dots,u_d) }.
    \end{align*}
    We have left to show that each of $A_3$, $A_4$ converges to $0$ as
    $n \to \infty$. The result for $A_3$ follows from~\eqref{eq_rate_lemma}.
    And for $A_4$, for any $(u_1, \dots u_d) \in [0,1]^d$, we know by the assumption
    $\Copula\in\MDA(\Copula^*)$ that
    $\Copula^n(u_1^{1/n},\dots,u_d^{1/n}) \to \Copula^*(u_1,\dots,u_d)$;
    by \cite[Theorem~1.7.6]{DuranteSempi2016}, this convergence is even uniform.
    Hence $A_4\to 0$ for $n\to\infty$ and we are done.
  \item[\ref{rate_function_convergence_II} $\Rightarrow$ \ref{rate_function_convergence_III}:]
    Clearly the set of all diverging rate functions includes those asymptotically equivalent to $n$.
  \item[\ref{rate_function_convergence_III} $\Rightarrow$ \ref{rate_function_convergence_I}:]
    Fix $\varepsilon > 0$ and define
    $\pD(u)=u^{\theta}$, where $\theta = \lim_{n\to\infty} n/\rate_n$.  Furthermore,
    let $\Copula_n(\bfu) := \Copula^n(u_1^{1/n},\dots,u_d^{1/n})$ and
    $\Copula_{\rate_n}(\bfu) :=
    \Copula^{\rate_n}(u_1^{1/\rate_n},\dots,u_d^{1/\rate_n})$.  Given that
    $\Copula^*$ is an EVC and $\pD(u)=u^{\theta}$ is a power distortion we have
    $\Copula^*_{\pD}(\bfu) = \Copula^*(\bfu)$ and hence the decomposition
    \begin{align*}
      \abs{\Copula_n(\bfu) - \Copula^*(\bfu)} &= \abs{
                                                \Copula_{n}(\bfu) - \Copula^*_{\pD}(\bfu)
                                                }
                                                \leq \abs{\Copula_{n}(\bfu) - \left(\Copula_{\rate_n}\right)_{\pD}(\bfu)
                                                } + \abs{\Copula^*_{\pD}(\bfu) - \left(\Copula_{\rate_n}\right)_{\pD}(\bfu)}.
    \end{align*}
    Concerning the first term on the right hand side we have by our assumption
    that $n/\rate_n\to\theta$ and hence that
    $\widetilde{\pD}_n(u) = u^{\theta \rate_n/n} \to u = \id(u)$, where this
    convergence is uniform since $\widetilde{\pD}_n$ is non-decreasing and the identity
    is continuous; see \cite[p.~1]{Resnick1987}.  There therefore exists a
    $N_1\in\IN$ such that, for all $n\geq N_1$, we have
    $\Vert \widetilde{\pD}_n - \id\Vert_{\infty} \leq \varepsilon/(2(d+1))$.
    Utilizing the alternative expression
    \begin{align*}
      (\Copula_{\rate_n})_{\pD}(\bfu) = \Copula^{n \theta \rate_n/n}\bigl(u_1^{1/(n \theta \rate_n/n)},\dots,u_d^{1/(n\theta\rate_n/n)}\bigr) = (\Copula_{n})_{\widetilde{\pD}_n}(\bfu),
    \end{align*}
    and that $\Copula_n$ is a copula for each $n\in\IN$, Proposition~\ref{prop_distorted_measure_difference}~\ref{prop_max_difference_I} now guarantees for $n\geq N_1$ that
    \begin{align*}
      \abs{\Copula_{n}(\bfu) - (\Copula_{\rate_n})_{\pD}(\bfu)} = \abs{\Copula_{n}(\bfu) - (\Copula_{n})_{\widetilde{\pD}_n}(\bfu)} \leq \varepsilon/2,
    \end{align*}
    where it is not necessary that $(\Copula_{n})_{\widetilde{\pD}_n}$ is a $d$-copula.

    Concerning the second term on the right hand side, first note that for
    $\theta \geq 1$ the power function $\pD$ is Lipschitz continuous on $[0,1]$
    with Lipschitz constant $c=\theta$, while for $\theta \in (0,1)$ the
    function is Hölder continuous with constant $c=1$ and exponent $\theta$.
    Therefore,
    \begin{align*}
      \abs{\left(\Copula_{\rate_n}\right)_{\pD}(\bfu)-\Copula^*_{\pD}(\bfu)} &= \abs{\pD\bigl(\Copula_{\rate_n}(\qD(\bfu))\bigr)-\pD\bigl(\Copula^*(\qD(\bfu))\bigr)}\\
                                                                             &\leq \begin{cases}
                                                                               \abs{\Copula_{\rate_n}(\qD(\bfu))-\Copula^*(\qD(\bfu))}^{\theta}, & 0 < \theta < 1; \\
                                                                               \theta\abs{\Copula_{\rate_n}(\qD(\bfu))-\Copula^*(\qD(\bfu))}, & \theta > 1.
      \end{cases}
    \end{align*}
    The assumed pointwise convergence of $\Copula_{\rate_n}$ to $\Copula^*$ now
    allows to pick $N_2\in\IN$ such that, for all $n\geq N_2$, one has
    $\abs{\Copula_{\rate_n}(\qD(\bfu))-\Copula^*(\qD(\bfu))} \leq \varepsilon'$,
    where we set $\varepsilon' = (\varepsilon/2)^{1/\theta}$ in case
    $\theta\in(0,1)$ and $\varepsilon' = \varepsilon/(2\theta)$ in case
    $\theta \geq 1$.  For $n \geq N_2$, we thus have that
    $\abs{\left(\Copula_{\rate_n}\right)_{\pD}(\bfu)-\Copula^*_{\pD}(\bfu)}\leq
    \varepsilon/2$, which establishes that for all $n\geq \max(N_1,N_2)$, we have
    $\abs{\Copula_n(\bfu) - \Copula^*(\bfu)} \leq \varepsilon/2+\varepsilon/2 =
    \varepsilon$ and hence that $\Copula_n$ converges pointwise to $\Copula^*$.
    By Definition~\ref{def_extreme_value_copula}, $\Copula\in\MDA(\Copula^*)$.
  \item[\ref{rate_function_convergence_IV} $\Leftrightarrow$ \ref{rate_function_convergence_I}:]
    See \cite[Théorème~(3.1), p.~14]{deheuvels1978}. \qedhere
  \end{itemize}
\end{proof}

\begin{proof}[Proof of Theorem~\ref{theorem_max_domain_of_attraction_Morillas}]
  \mbox{}
  \begin{enumerate}
  \item With
    $\pD\in\MDA(\EW_{\gamma})$ and $\support_+(\pD) = 1$ we have from
    \cite[Theorem~3.1~(e), p.~636]{MohanRavi1993} in combination with
    \cite[Theorem~2.2, p.~634]{MohanRavi1993} for $u\in[0,1]$ that
    \begin{align*}
      \lim_{n\to\infty}\pD^n\left(u^{1/\rate_n}\right) = \e^{-(-\log u)^{\gamma}},
    \end{align*}
    where the associated rate is given as
    $\rate_n := \left(-\log\qD(1-1/n)\right)^{-1}$.
    Furthermore, $\rate_n\to\infty$ for $n\to\infty$.
    Now, if the Morillas-type distorted
    $\Copula_{\pD}$ is a $d$-copula, consider the relevant sequence
    \begin{align*}
      \Copula_{\pD}^n(u_1^{1/n},\dots,u_d^{1/n}) = \pD^n\left(\Copula\left(\qD\left( u_1^{1/n} \right), \ldots, \qD\left( u_d^{1/n} \right)\right)\right).
    \end{align*}
    Using the powers $\rate_n$ and $1/\rate_n$, we can rewrite this transformation as
    \begin{align*}
      \Copula_{\pD}^n(u_1^{1/n},\dots,u_d^{1/n}) =
      \pD^n\left( \left(\Copula^{\rate_n}\left( \left( \left(\qD\left( u_1^{1/n} \right)\right)^{\rate_n}\right)^{1/\rate_n}, \ldots, \left(\left(\qD\left( u_d^{1/n} \right)\right)^{\rate_n}\right)^{1/\rate_n} \right) \right)^{1/\rate_n} \right).
    \end{align*}
    Let
    $\Copula_{\rate_n}(u_1,\dots,u_d) := \Copula^{\rate_n}(u_1^{1/\rate_n}, \ldots,
    u_d^{1/\rate_n})$ and $\pD_n(u) := (\pD(u^{1/\rate_n}))^n$. Then
    \begin{align*}
      \Copula_{\pD}^n(u_1^{1/n},\dots,u_d^{1/n}) = \pD_n\left(\Copula_{\rate_n}\left(\pD_n^{-1}(u_1), \ldots, \pD_n^{-1}(u_d)\right)\right).
    \end{align*}
    As argued above, the pointwise convergence $\pD_n(u) \to \pD^*(u)$ holds and
    that $\pD_n$ is non-decreasing.
    Furthermore, \cite[Proposition~0.1, p.~5]{Resnick1987}
    guarantees that $\qD_n(u) \to \qDs(u)$.  For
    $g_n \colon [0,1]^d\to[0,1]^d$ defined by
    $g_n(u_1,\dots,u_d):=(\pD_n^{-1}(u_1),\dots,\pD_n^{-1}(u_d))$ we then have the
    pointwise convergence $g_n(u_1,\dots,u_d) \to g(u_1,\dots,u_d)$, where
    $g(u_1,\dots,u_d) = (\qDs(u_1),\dots,\qDs(u_d))$. By
    Lemma~\ref{rate_function_convergence}~\ref{rate_function_convergence_II}, we have that
    $\Copula_{\rate_n}(u_1,\dots,u_d)\to\Copula^*(u_1,\dots,u_d)$, and this convergence is
    even uniform; see \cite[Theorem~1.7.6, p.~28]{DuranteSempi2016}. By
    Lemma~\ref{composition_convergence}, we thus have the pointwise
    convergence
    \begin{align*}
      \Copula_{\rate_n}\bigl(\pD_n^{-1}(u_1), \ldots, \pD_n^{-1}(u_d)\bigr) = (\Copula_{\rate_n}\circ g_n)(u_1,\dots,u_d) \to \Copula^*\bigl(\qDs(u_1), \ldots, \qDs(u_d)\bigr).
    \end{align*}
    Noting further that $\pD_n$ is non-decreasing and that the limit function $\pD^*$
    is continuous, the convergence $\pD_n(u)\to\pD^*(u)$ is in fact
    uniform; see \cite[p.~1]{Resnick1987}. Again applying Lemma~\ref{composition_convergence} thus yields
    \begin{align*}
      \pD_n\left(\Copula_{\rate_n}\left(\pD_n^{-1}(u_1), \ldots, \pD_n^{-1}(u_d)\right)\right) \to \pD^*\left(\Copula^*\left(\qDs(u_1), \ldots, \qDs(u_d)\right)\right),
    \end{align*}
    that is we have the pointwise convergence
    $\Copula_{\pD}^n(u_1^{1/n},\dots,u_d^{1/n}) \to
    \Copula^*_{\pD^*}(u_1,\dots,u_d)$.  Given that
    $\Copula_{\pD}^n(u_1^{1/n},\dots,u_d^{1/n})$ is a Morillas-type distorted
    $\Copula_{\pD}$ based on power distortions, it is a copula for $n\geq d$ since
    the power distortion $u^n$ is in this case $d$-monotone.  Consequently,
    \cite[Theorem~1.7.5, p.~28]{DuranteSempi2016} implies that the limit
    $\Copula^*_{\pD^*}$ is also a copula.  Thus, $\Copula^*_{\pD^*}$ is
    identified as an EVC, see Definition~\ref{def_extreme_value_copula}, and
    $\Copula_{\pD}\in\MDA(\Copula^*_{\pD^*})$. %
  \item $\pD\in\AMDCO{d}$
    implies that the conditions for~\ref{item_distortion_theorem_1} are satisfied, where we specifically have $\support_+(\pD)=1$ since we otherwise contradict the convexity of $\pD$.
    To see that
      $\pD^*(u) = \e^{-(-\log u)^{\gamma}}$ with the additional restriction
      $\gamma \in (0,1]$, a first step is to notice that the results so far are
      invariant under the dependence structure in the following sense: For a fixed
      $\pD\in\AMDCO{d}$ such that
      $\lim_{n\to\infty}\pD^{n}(u^{1/{\rate_n}}) = \pD^*(u)$, we have
      $\Copula_{\pD}\in\MDA(\Copula^*_{\pD^*})$ for any
      $\Copula\in\MDA(\Copula^*)$.  Furthermore, via
      Definition~\ref{def_extreme_value_copula}, we have
      $\Copula^*\in\SetEVTCopula_d$ if and only if there exists at least one
      copula $\Copula$ such that $\Copula\in\MDA(\Copula^*)$.  As such, we
      have established that
      $\Copula\in\MDA(\Copula^*)$ implies that $\Copula_{\pD}\in\MDA(\Copula^*_{\pD^*})$ for any EVC
      $\Copula^*\in\SetEVTCopula_d$.  This means that
      $\pD^*$ transforms \emph{any} $d$-dimensional EVC into a $d$-dimensional
      EVC, and therefore Corollary~\ref{cor_EVC_to_EVC} implies that
      $\pD^*$ must be of the form $\pD^*(u) = \exp(-\lambda(-\log u)^{\gamma})$
      for some $\lambda > 0$ and $\gamma \in (0,1]$. Considering the previously derived form of $\pD^*$ this clearly implies $\lambda=1$ and confirms $\pD^*(u) = \e^{-(-\log u)^{\gamma}}$ with $\gamma \in (0,1]$,

    Finally, the claim that $\Copula^*_{\pD^*} = \Copula^*$ when $\gamma=1$ follows trivially since in this case $\pD^*(u)=u$.
  \end{enumerate}
\end{proof}

\begin{proof}[Proof of Corollary~\ref{cor_d_abs_monotone_not_Gumbel}]
  By assumption we have for $u\in(1-\varepsilon,1)$ that $1-\pD(u) \geq 1-u$, so $\log(1-\pD(u)) \geq \log(1-u)$ and thus
  \begin{align*}
    \frac{\log(1-\pD(u))}{\log(1-u)} \leq 1,\quad u\in(1-\varepsilon,1).
  \end{align*}
  \cite[Corollary~2.5.3]{deHaan1970} now directly implies the statement.
  The particular case of $\pD\in\AMDCO{d}$ follows since the implied convexity
  forces any such $\pD$ to satisfy $\pD(u) \leq u$ for $u\in[0,1]$.
\end{proof}

\begin{proof}[Proof of Proposition~\ref{prop_D_in_Weibull_MDA}]
  To verify the Weibull MDA condition, see
  \cite[Proposition~1.13, p.~59]{Resnick1987}, let $t>0$ be fixed and use
  l'Hôpital's rule in conjunction with $0 < \dD(1) < \infty$ to get
  \begin{align*}
    \lim_{x \to \infty}
    \frac{1 - \pD\bigl(1-1/(tx)\bigr)}{1 - \pD\bigl(1-1/x\bigr)}
    =
    \lim_{x \to \infty}
    \frac{x^2 \dD\bigl(1-1/(tx)\bigr)}{tx^2 \dD\bigl(1-1/x\bigr)}
    =
    \frac{1}{t} \frac{\dD(1)}{\dD(1)} = t^{-1},
  \end{align*}
  which implies that $D \in \MDA(\Psi_1)$.
  The statement for $\pD\in\AMDCO{d}$ follows since convexity of $\pD$ implies
  that the derivative $\dD$ of $\pD$ is non-decreasing, that there is
  necessarily a point $u_0 \in (0,1)$ such that $\dD(u_0) > 0$ (otherwise $\pD$
  would be constant equal to zero) and therefore
  $\dD(1) = \lim_{u\nearrow 1}\dD(u)>0$, while $\dD(1) < \infty$ by
  assumption.
\end{proof}

\begin{proof}[Proof of Corollary~\ref{corollary_abs_monotone_domain_of_attraction}]
  As a power series, $\pD $ is differentiable on $(0,1)$ with derivative
  $\dD(u) = \sum_{n = 1}^\infty n \Prob(N = n) u^{n-1} = \sum_{n = 0}^\infty
  (n+1) \Prob(N = n+1) u^{n}$.  Given that
  $0 < \E[N] = \sum_{n = 0}^\infty (n+1) \Prob(N = n+1) < \infty$ we can apply
  Abel's theorem, see \cite[8.2~Theorem, p.~174]{Rudin1976}, to get
  $\lim_{u \nearrow 1} \dD(u) = \dD(1) = \E[N]$, that is $\dD$ is continuous at $1$
  and also non-zero close to $1$. The result now follows from
  Proposition~\ref{prop_D_in_Weibull_MDA}.
\end{proof}

\begin{proof}[Proof of Lemma~\ref{lemma_Archimedean_D_regular_variation}]
  \mbox{}
  \begin{enumerate}
  \item Consider the representation
  \begin{align*}
    1-\pD_{\gen}(1-1/x)
    &= 1 - \gen\left(-\log\left(1-1/x\right)\right)
      =\left(\left(1 - \gen\left(1/\cdot\right)\right)\circ \left(-\log\left(1-1/\cdot\right)\right)^{-1}\right)(x).
  \end{align*}
  It is easily verified that
  $\left(-\log\left(1-1/\cdot\right)\right)^{-1}\in\RV_{1}$ and hence, as
  composition of regularly varying functions (see \cite[Proposition~1.5.7~(ii),
  p.~26]{BinghamGoldieTeugels1987}), this implies
  $1-\pD_{\gen}(1-1/\cdot)\in\RV_{-\rho}$.
  \item Consider $\gen_{\pD}(x) := \pD(\e^{-x})$ and write
  \begin{align*}
    1-\gen_{\pD}(1/x)
    &= 1 - \pD\left(1 - \frac{1}{1/\left(1-\e^{-1/x}\right)}\right)
      = \left(1 - \pD\left(1 - \frac{1}{\cdot}\right)\right)\circ(1-\e^{-1/\cdot})^{-1} (x),
  \end{align*}
  where one can verify that $(1-\e^{-1/\cdot})^{-1}\in\RV_1$.
  By assumption, $1 - \pD(1 - \frac{1}{\cdot})\in\RV{-\gamma}$. Composition of regularly varying functions (see \cite[Proposition~1.5.7~(ii), p.~26]{BinghamGoldieTeugels1987}) thus shows that $1-\gen_{\pD}(1/\cdot)\in\RV_{-\gamma}$. \qedhere
  \end{enumerate}
\end{proof}

\begin{proof}[Proof of Corollary~\ref{cor_Archimedean_max_domain}]
  Based on the assumptions, Lemma~\ref{lemma_Archimedean_D_regular_variation} implies that
  $\pD_{\gen}=\gen\circ(-\log)\in\MDA(\Psi_{\rho})$.
Utilizing the rate function $\rate_n$ associated to $\pD_{\gen}$ as given in \eqref{eq_rate} we then have $\pD_{\gen}^n(u^{1/\rate_n})\to \e^{-(-\log u)^{\rho}}$. As
  discussed in \cite[p.~170]{morillas2005b}, we have
  $\Copula_{\gen} = \Pi_{\pD_{\gen}}$, which implies via
  Theorem~\ref{theorem_max_domain_of_attraction_Morillas}~\ref{item_distortion_theorem_1}
  that $\Copula_{\gen}\in\MDA(\Pi_{\e^{-(-\log \cdot)^{\rho}}})$.  A direct
  calculation then shows that $\Pi_{\e^{-(-\log \cdot)^{\rho}}} = \GHCopula_{1/\rho}$.
\end{proof}

\end{document}

%
%
%
%

%%% Local Variables:
%%% mode: LaTeX
%%% TeX-master: t
%%% End: